\documentclass[11pt]{article}

\usepackage[a4paper,margin=1in]{geometry}
\usepackage{amsmath,amssymb,amsthm,mathtools}
\usepackage{booktabs,array}
\usepackage{enumitem}
\usepackage{xcolor}
\usepackage{tikz}
\usetikzlibrary{backgrounds}
\usepackage[hidelinks]{hyperref}

\usepackage{thmtools}
\usepackage{thm-restate}

\declaretheorem[numberwithin=section]{theorem}
\declaretheorem[sibling=theorem]{lemma}
\declaretheorem[sibling=theorem]{proposition}

\declaretheorem[sibling=theorem]{fact}
\declaretheorem[sibling=theorem,style=definition]{definition}

\declaretheorem[sibling=theorem,style=definition]{example}

\newcommand{\clonefont}[1]{\mathsf{#1}}
\newcommand{\BF}{\clonefont{BF}}
\newcommand{\Rc}{\clonefont{R}}
\newcommand{\Mc}{\clonefont{M}}
\newcommand{\Sc}{\clonefont{S}}
\newcommand{\Dc}{\clonefont{D}}
\newcommand{\Lc}{\clonefont{L}}
\newcommand{\Ec}{\clonefont{E}}
\newcommand{\Vc}{\clonefont{V}}
\newcommand{\Nc}{\clonefont{N}}
\newcommand{\Ic}{\clonefont{I}}
\newcommand{\Cc}{\clonefont{C}}
\newcommand{\Kc}{\clonefont{K}}

\newcommand{\PL}{\mathsf{PL}}
\newcommand{\CIR}{\mathsf{CIR}}
\newcommand{\CNF}{\mathsf{CNF}}
\newcommand{\exCNF}{\exists\mathsf{CNF}}
\newcommand{\opt}{\mathsf{opt}}
\newcommand{\Ptime}{\mathsf{P}}
\newcommand{\NP}{\mathsf{NP}}
\newcommand{\RP}{\mathsf{RP}}
\newcommand{\zero}{\mathbf 0}
\newcommand{\one}{\mathbf 1}
\newcommand{\Ftwo}{\mathbb F_2}
\newcommand{\maj}{\operatorname{maj}}
\newcommand{\thr}{\mathrm{th}}
\newcommand{\CSP}{\mathsf{CSP}}
\newcommand{\Pol}{\operatorname{Pol}}
\newcommand{\Inv}{\operatorname{Inv}}
\newcommand{\size}[1]{\operatorname{size}(#1)}
\newcommand{\dep}{\operatorname{depth}}
\newcommand{\pwm}{\le_{\mathrm{pwm}}}

\title{Fitting and Learning Basis-Restricted Propositional Formulas}
\author{Balder ten Cate}

\begin{document}
\maketitle

\begin{abstract}
For a finite set $O$ of Boolean functions, we consider the class of
propositional formulas built using the functions in $O$ as connectives. 
We determine, for each possible choice of $O$,
the complexity of various fitting and learning problems.
These include: finding a formula that fits a given labeled sample, finding a
small one (an Occam algorithm), minimizing the number of misclassified
examples when the sample is not realizable (empirical risk minimization),
and several forms of PAC learning.  Our results apply both to formulas (represented as trees)
and to circuits. We also briefly discuss the status of
the same questions for other kinds of propositional fragments.
\end{abstract}

\section{Introduction}

Fix a finite set $O$ of Boolean functions and let $\PL_O$ be the class of
propositional formulas that can be built using the functions in $O$ as connectives.  Varying $O$ gives a family
of concept classes, and a natural
question is how the difficulty of the standard learning-theoretic tasks depends
on the choice of $O$. 
Several such classifications are either known or implied by known results, but they are scattered across the literature.
We put them side by
side, focusing on four problems.  
\begin{itemize}
\item \emph{Fitting.}  Given a labeled sample, decide whether some
$\PL_O$-formula agrees with all of it, and if so produce one.
\item \emph{Occam algorithms.}  If the sample is realizable, find a fitting
formula of near-minimal size.
\item \emph{Empirical risk minimization.}  When the sample need not be realizable,
find a $\PL_O$-formula minimizing the number of mistakes.
\item \emph{PAC learning.}  Is $\PL_O$ efficiently learnable in (variants of) the
PAC model?
\end{itemize}

The four problems are not independent.  Fitting is the special case of
empirical risk minimization in which the sample is realizable, so an ERM
algorithm is in particular a fitting algorithm.  An Occam algorithm is a fitting
algorithm that compresses, and compression yields a PAC learner, so hardness of
PAC learning rules out Occam algorithms as well.  

These questions can be studied also for $\CIR_O$, that is,
the family of circuits using functions from $O$ as gates. Note that $\PL_O$ and $\CIR_O$ have the
same expressive power, but may differ in succinctness. In addition, in 
Section~\ref{sec:constraints} we discuss other fragments, such as monotone CNF and
Horn CNF, that do not fall under the above basis-generated regime.

Our contributions are partly organizational: we give a uniform presentation and we provide 
written-out proofs for results that are only sketched in the literature. Our main novel 
contributions are an algorithm for constructing fitting formulas efficiently
based on a refinement of the classic Baker--Pixley construction, and a 
trichotomy theorem for empirical risk minimization, building on known results for 
hypergraph vertex-cover problems. We also fill a small
gap, showing that the PAC-learnability dichotomy in~\cite{Dalmau} holds not
only for circuits but also for formulas.

The picture that arises is as follows.

\paragraph{Fitting and Occam algorithms.}
The basic fitting problem turns out to be uniformly easy:

\begin{restatable}[Fitting]{theorem}{thmfitting}\label{thm:boolean-fitting-constructive}
Let $O$ be a fixed finite basis.  Then $\PL_O$ has a polynomial-time fitting
algorithm: given a labeled sample it decides whether the sample is realizable,
and if so returns a fitting $\PL_O$-formula of polynomial size.  The same
holds for $\CIR_O$.
\end{restatable}

Here, by a \emph{labeled sample} we mean a finite list of \emph{labeled
examples} $(\mathbf a,b)$ with $\mathbf a\in\{0,1\}^n$ and $b\in\{0,1\}$; a
formula \emph{fits} the sample if it takes the value $b$ at $\mathbf a$ for
each of its examples, and the sample is \emph{realizable} for $\PL_O$ if some
$\PL_O$-formula fits it.
The tractability of testing existence
 was already observed in~\cite{Mayr}. We give a proof in Section~\ref{sec:fitting}, 
 which also shows
 how to construct fitting formulas efficiently, using a refinement of the
 Baker--Pixley theorem.

One may ask for a \emph{short} fitting
formula.  An
\emph{Occam algorithm} for $\PL_O$ is a polynomial-time fitting algorithm that,
on every realizable sample, returns a fitting formula of size at most
\(
  p\bigl(s_{\mathrm{opt}},\,n\bigr)\cdot m^{\beta}
\)
for some fixed polynomial $p$ and some fixed $\beta<1$,
where $m$ is the number of examples, $n$ is the number of Boolean variables,
and $s_{\mathrm{opt}}$ is the least size of a fitting $\PL_O$-formula.
An Occam algorithm compresses, as it returns something sublinear in $m$.
One may ask for more, namely an \emph{attribute-efficient} Occam
algorithm, whose size bound does not depend on $n$: a fitting formula of size at
most $p(s_{\mathrm{opt}})\cdot m^{\beta}$.

\begin{restatable}[Occam algorithms, from~\cite{Dalmau}]{theorem}{thmoccam}\label{thm:occam}
Let $O$ be a finite basis of Boolean functions.
\begin{enumerate}[label=(\alph*)]
\item If $O\preceq\{\wedge,\top,\bot\}$, $O\preceq\{\vee,\top,\bot\}$ or
$O\preceq\{\neg,\top,\bot\}$, then $\PL_O$ has an attribute-efficient Occam
algorithm.
\item If $\{\oplus^3\}\preceq O\preceq\{\oplus,\top,\bot\}$, then $\PL_O$ has an
Occam algorithm, but we do not know whether it has an attribute-efficient one.
\item Otherwise $\PL_O$ does not have an Occam algorithm, under the
cryptographic assumptions of Section~\ref{sec:pac}.
\end{enumerate}
The same holds for $\CIR_O$.
\end{restatable}

The proof is given in Section~\ref{sec:occam}; the negative part is a
consequence of the classification of PAC learnability,
Theorem~\ref{thm:pac} below, since an Occam algorithm yields a proper PAC
learner~\cite{Blumer}.

\paragraph{Empirical risk minimization.}
The next problem drops the assumption that the sample is realizable
and asks
for a hypothesis that minimizes the empirical risk, that is, the fraction of misclassified examples.
This is empirical risk minimization (ERM).
Unlike the fitting problem, ERM is not always solvable in polynomial time.
When it is hard, one may ask for an approximation algorithm. We say that an algorithm
is a \emph{weak approximator} for $\PL_O$ if there are $\delta,\varepsilon>0$
such that, on every sample whose optimal hypothesis has empirical risk
$\le\varepsilon$, the algorithm produces a $\PL_O$-hypothesis of empirical risk
$\le1/2-\delta$. This is, in some sense, the lowest bar. Note that,
if $O$ includes the truth constants $\top$ and $\bot$, then there is always
a trivial solution that achieves an error fraction $\le 1/2$.
Every $k$-approximation algorithm for a constant factor $k$ is a weak approximator:
it suffices to pick
$\varepsilon=1/(4k)$  and $\delta=1/4$.

To state the classification, for $r>j\ge1$ let $\thr^r_j$ denote the
$r$-ary threshold operation ``at least $j$ of $r$'', that is,
\[
  \thr^r_j(x_1,\ldots,x_r)=1
  \quad\Longleftrightarrow\quad
  x_1+\cdots+x_r\ge j .
\]
Thus $\thr^r_j$ and $\thr^r_{r-j+1}$ are dual, and $\thr^3_2=\maj$.
Furthermore, 
in what follows, $O\preceq O'$ means that every function in $O$ is
term-definable from $O'$, and $O\equiv O'$ that the two are interdefinable.

\begin{restatable}[Empirical risk minimization]{theorem}{thmerm}\label{thm:erm}
Let $O$ be a finite basis of Boolean functions.  The ERM problem for $\PL_O$
falls into one of the following three regimes.

\begin{enumerate}[label=(\roman*)]
\item If
\[
\begin{gathered}
  \{\wedge\}\preceq O\preceq\{\wedge,\top,\bot\},
  \qquad
  \{\vee\}\preceq O\preceq\{\vee,\top,\bot\}, \qquad
  \text{or} \qquad
  \{\oplus^3\}\preceq O\preceq\{\oplus,\top,\bot\},
\end{gathered}
\]
then, unless $\Ptime=\NP$, there is no polynomial-time weak approximator for
$\PL_O$, whatever its constants $\delta,\varepsilon>0$.  In particular ERM for
$\PL_O$ is $\NP$-hard and has no polynomial-time constant-factor
approximation.
\item If, for some $k\ge2$,
\[
  \{\thr^{k+1}_2\}\preceq O\preceq\{\to,\ \thr^{k+1}_2\}
  \qquad\text{or}\qquad
  \{\thr^{k+1}_k\}\preceq O\preceq\{x\wedge\neg y,\ \thr^{k+1}_k\},
\]
then ERM for $\PL_O$ is $\NP$-hard; it admits a polynomial-time
$k$-approximation; and, under the Unique Games Conjecture, it admits no
polynomial-time $(k-\varepsilon)$-approximation for any $\varepsilon>0$.  The
factor $k$ is therefore optimal under that conjecture.
\item In all other cases, ERM for $\PL_O$ is solvable in polynomial time.
\end{enumerate}
The same holds for $\CIR_O$.
\end{restatable}

The proof is given in Section~\ref{sec:erm-regions}.  Note that, in light of
Theorem~\ref{thm:boolean-fitting-constructive}, ERM may  be equivalently
viewed as the problem, given a labeled sample,
of finding a relabeling that is realizable and disagrees on as few examples
as possible.

\paragraph{Learning from random examples.}
We now turn to learning, beginning with the combinatorial parameter that
governs how many examples are needed.

\begin{restatable}[VC dimension]{proposition}{thmvc}\label{prop:vc}
Let $O$ be a finite basis of Boolean functions and let $n\ge1$.  If
$O\preceq\{\wedge,\top,\bot\}$, $O\preceq\{\vee,\top,\bot\}$ or
$O\preceq\{\oplus,\top,\bot\}$, then the VC dimension of $\PL_O$ in $n$
variables is at most $n+1$.  Otherwise it is $2^{\Omega(n)}$.
\end{restatable}

We now consider two notions of learnability.  Both are
representation-sensitive: the learner is
told a bound $s$ on the size of a representation of the target, and may use a
number of examples polynomial in $1/\varepsilon$, $1/\delta$, $n$ and $s$.  By a
\emph{proper PAC learner} for $\PL_O$ we mean an algorithm that, given
$\varepsilon,\delta>0$, the bound $s$, and sufficiently many examples drawn from an arbitrary
distribution and labeled by a target in $\PL_O$ of size at most $s$, outputs in
polynomial time a hypothesis from the class whose error is at most $\varepsilon$ with
probability at least $1-\delta$.  A \emph{PAC predictor} need not output a
hypothesis: after seeing the examples it is given one further point and must
predict its label, with error bounded away from $1/2$ by an inverse polynomial,
and it may in addition ask \emph{membership queries}, that is, ask for the value
of the target at points of its own choosing.  PAC prediction is easier than
producing a hypothesis, and membership queries only help, so hardness of
PAC prediction with queries is the stronger statement.\footnote{This relies on the 
fact that, in the concept classes we consider here,  the label of a given example
for a given concept can be computed in polynomial time.} 
 Both notions, and the
cryptographic assumptions, are given precisely in
Section~\ref{sec:pac}.

\begin{restatable}[{PAC learning and PAC prediction~\cite{Dalmau}}]{theorem}{thmpac}\label{thm:pac}
Let $O$ be a finite basis.  The following are equivalent.
\begin{enumerate}[label=(\alph*)]
\item $O\preceq\{\wedge,\top,\bot\}$, $O\preceq\{\vee,\top,\bot\}$ or
      $O\preceq\{\oplus,\top,\bot\}$;
\item $\PL_O$ is polynomially properly PAC learnable;
\item $\PL_O$ is polynomially PAC predictable with membership queries.
\end{enumerate}
The implications (a)~$\Rightarrow$~(b)~$\Rightarrow$~(c) are unconditional; the
remaining implication (c)~$\Rightarrow$~(a) holds under cryptographic
assumptions described in Section~\ref{sec:pac}.  The same equivalence holds with
$\CIR_O$ in place of $\PL_O$.
\end{restatable}

Dalmau~\cite{Dalmau} states the above result for formulas and for circuits but the proof establishes it 
only for circuits. We supply
the missing step in Section~\ref{sec:pac}.  Other models of learning are
considered in~\cite{Dalmau} as well, namely exact learnability from membership
and equivalence queries.  We omit them here, but note that our proof shows
that the corresponding results of~\cite{Dalmau} also hold for formulas, and
not just for circuits.

A finer question is whether, in the positive cases, the number of examples
can be made \emph{attribute-efficient}~\cite{Littlestone88}: polynomial in
$s$, $1/\varepsilon$, $1/\delta$ and $\log n$ rather than in $n$ (note that
$\log n$ is the information theoretic content in bits of a single propositional variable).  For the
conjunctive and disjunctive bases this holds, as follows from Theorem~\ref{thm:occam} above.
For the affine bases the target is a parity of at most $s$ variables, and
whether such parities can be learned attribute-efficiently in polynomial
time is a well-known open problem~\cite{KlivansServedio,BshoutyHaddad}.

The above results also yield a complete picture for \emph{agnostic
learning}~\cite{KSS}, where we don't assume that the labels are given by a
target in the class: here, the examples are drawn from an arbitrary
distribution $D$ on $\{0,1\}^n\times\{0,1\}$, and the learner must output
a hypothesis whose error under $D$ exceeds
the least error $\opt_D$ of any $n$-ary $\PL_O$-formula by at most
$\varepsilon$. Note that, since there is no target, the bound for the
running time and sample size is in $n$,
$1/\varepsilon$ and $1/\delta$.
The learner is \emph{proper} if the hypothesis is a $\PL_O$-formula, and it
is a \emph{weak agnostic learner} if it is required to work only when
$\opt_D\le\varepsilon$, and then only to reach error $\le1/2-\delta$, for
some fixed $\varepsilon,\delta>0$.  Every weak agnostic learner yields a randomized weak
approximator for ERM, by running it on the uniform distribution over a given
sample.  Conversely, when the VC dimension of $\PL_O$ is polynomial in
$n$, every polynomial-time ERM algorithm yields a proper agnostic learner,
by uniform convergence; and polynomial VC dimension is in fact necessary
for weak agnostic learning, even improper~\cite{SSBD}.  It follows from Theorem~\ref{thm:erm}
and Proposition~\ref{prop:vc} that $\PL_O$ is properly agnostically
learnable in polynomial time when $O\preceq\{\neg,\top,\bot\}$; that in the
three intervals of regime~(i) of Theorem~\ref{thm:erm} it is \emph{not}
properly weak agnostically learnable in polynomial time, unless
$\NP=\RP$; and that for every other basis it is not weak agnostically
learnable at all, properly or not, for VC dimension reasons
(cf.~Proposition~\ref{prop:vc}).

Finally, a further variant of the PAC model deserves discussion.  Under \emph{random
classification noise} each label shown to the learner is flipped independently
with a fixed probability $\eta<1/2$~\cite{AngluinLaird}.  The usual route to
learning algorithms that are tolerant to such noise is Kearns's \emph{statistical query} model, in which the
learner sees no examples at all but may ask for the probability, under the
example distribution, of any polynomial-time predicate of an example and its
label, and receives it to within a tolerance of its choosing.  A class learnable
from polynomially many such queries of inverse-polynomial tolerance is PAC
learnable under random classification noise of any rate
$\eta<1/2$~\cite{Kearns98}.  Both notions are made precise in
Section~\ref{sec:noise}.

\begin{restatable}[Statistical queries, from~\cite{Kearns98,BFJKMR}]{theorem}{thmsq}\label{thm:sq}
Let $O$ be a finite basis.  Under the cryptographic assumptions of
Section~\ref{sec:pac}, $\PL_O$ is efficiently learnable from
statistical queries if and only if $O\preceq\{\wedge,\top,\bot\}$,
$O\preceq\{\vee,\top,\bot\}$ or $O\preceq\{\neg,\top,\bot\}$.  The same holds
for $\CIR_O$.
\end{restatable}

For random classification noise itself the picture is incomplete. Specifically,
for the affine cases
$\{\oplus^3\}\preceq O\preceq\{\oplus,\top,\bot\}$ the question is open --- it
is known as the \emph{learning parity with noise} problem~\cite{BKW}.

Figure~\ref{fig:post-lattice} collects these classifications. It depicts
Post's lattice. Its elements are all Boolean clones, which we can think of as the
equivalence classes of the pre-order $\preceq$. They are colored to reflect the status
of the clones in question with respect to the classifications above.

In this paper we focus on problems where a formula is to be derived from
examples.  For the converse direction, where a formula is given and a
meaningful sample is to be generated for it --- for instance one that
characterizes the formula up to equivalence within the fragment --- dichotomy
results over Post's lattice exist as well, cf.~\cite{BtCGM}.

\begin{figure}[p]
\centering
\begin{tikzpicture}[
  x=0.85cm, y=0.445cm,
  clone/.style={circle, minimum size=6.0mm, inner sep=0pt, font=\small,
                line width=0.5pt, draw=black!55, fill=white},
  regI/.style  ={draw={rgb,255:red,22;  green,163; blue,74},  fill={rgb,255:red,220; green,252; blue,231}},
  regII/.style ={draw={rgb,255:red,37;  green,99;  blue,235}, fill={rgb,255:red,219; green,234; blue,254}},
  regIII/.style={draw={rgb,255:red,124; green,58;  blue,237}, fill={rgb,255:red,237; green,233; blue,254}},
  regIV/.style={draw={rgb,255:red,217; green,119; blue,6},   fill={rgb,255:red,254; green,243; blue,199}},
  regV/.style ={draw={rgb,255:red,225; green,29;  blue,72},  fill={rgb,255:red,255; green,228; blue,230}},
  edge/.style={draw=black!55, line width=0.4pt},
  displayed gap/.style={edge, dashed, dash pattern=on 2.4pt off 1.6pt},
  swatch/.style={circle, minimum size=3.4mm, inner sep=0pt, line width=0.5pt}
]
  \node[clone,regV] (BF) at (0.000,36.000) {$\BF$};
  \node[clone,regV] (R1) at (-1.200,34.500) {$\Rc_1$};
  \node[clone,regV] (R0) at (1.200,34.500) {$\Rc_0$};
  \node[clone,regV] (R2) at (0.000,33.000) {$\Rc_2$};
  \node[clone,regV] (M) at (0.000,30.000) {$\Mc$};
  \node[clone,regV] (M1) at (-1.200,28.500) {$\Mc_1$};
  \node[clone,regV] (M0) at (1.200,28.500) {$\Mc_0$};
  \node[clone,regV] (M2) at (0.000,27.000) {$\Mc_2$};
  \node[clone,regIV] (S21) at (8.040,26.700) {$\Sc_1^2$};
  \node[clone,regIV] (S31) at (8.040,24.960) {$\Sc_1^3$};
  \node[clone,regV] (S1) at (8.040,20.190) {$\Sc_1$};
  \node[clone,regIV] (S212) at (6.960,24.750) {$\Sc_{12}^2$};
  \node[clone,regIV] (S312) at (6.960,23.010) {$\Sc_{12}^3$};
  \node[clone,regV] (S12) at (6.960,18.240) {$\Sc_{12}$};
  \node[clone,regIV] (S211) at (5.760,24.750) {$\Sc_{11}^2$};
  \node[clone,regIV] (S311) at (5.760,23.010) {$\Sc_{11}^3$};
  \node[clone,regV] (S11) at (5.760,18.240) {$\Sc_{11}$};
  \node[clone,regIV] (S210) at (4.680,22.800) {$\Sc_{10}^2$};
  \node[clone,regIV] (S310) at (4.680,21.060) {$\Sc_{10}^3$};
  \node[clone,regV] (S10) at (4.680,16.290) {$\Sc_{10}$};
  \node[clone,regIV] (S20) at (-8.040,26.700) {$\Sc_0^2$};
  \node[clone,regIV] (S30) at (-8.040,24.960) {$\Sc_0^3$};
  \node[clone,regV] (S0) at (-8.040,20.190) {$\Sc_0$};
  \node[clone,regIV] (S202) at (-6.960,24.750) {$\Sc_{02}^2$};
  \node[clone,regIV] (S302) at (-6.960,23.010) {$\Sc_{02}^3$};
  \node[clone,regV] (S02) at (-6.960,18.240) {$\Sc_{02}$};
  \node[clone,regIV] (S201) at (-5.760,24.750) {$\Sc_{01}^2$};
  \node[clone,regIV] (S301) at (-5.760,23.010) {$\Sc_{01}^3$};
  \node[clone,regV] (S01) at (-5.760,18.240) {$\Sc_{01}$};
  \node[clone,regIV] (S200) at (-4.680,22.800) {$\Sc_{00}^2$};
  \node[clone,regIV] (S300) at (-4.680,21.060) {$\Sc_{00}^3$};
  \node[clone,regV] (S00) at (-4.680,16.290) {$\Sc_{00}$};
  \node[clone,regV] (D) at (0.000,21.000) {$\Dc$};
  \node[clone,regV] (D1) at (0.000,19.500) {$\Dc_1$};
  \node[clone,regIV] (D2) at (0.000,18.000) {$\Dc_2$};
  \node[clone,regII] (E) at (3.600,15.000) {$\Ec$};
  \node[clone,regII] (E1) at (2.400,13.500) {$\Ec_1$};
  \node[clone,regII] (E0) at (4.800,13.500) {$\Ec_0$};
  \node[clone,regII] (E2) at (3.600,12.000) {$\Ec_2$};
  \node[clone,regII] (V) at (-3.600,15.000) {$\Vc$};
  \node[clone,regII] (V0) at (-2.400,13.500) {$\Vc_0$};
  \node[clone,regII] (V1) at (-4.800,13.500) {$\Vc_1$};
  \node[clone,regII] (V2) at (-3.600,12.000) {$\Vc_2$};
  \node[clone,regIII] (L) at (0.000,15.000) {$\Lc$};
  \node[clone,regIII] (L0) at (1.200,13.500) {$\Lc_0$};
  \node[clone,regIII] (L1) at (-1.200,13.500) {$\Lc_1$};
  \node[clone,regIII] (L3) at (0.000,13.500) {$\Lc_3$};
  \node[clone,regIII] (L2) at (0.000,12.000) {$\Lc_2$};
  \node[clone,regI] (N) at (0.000,9.750) {$\Nc$};
  \node[clone,regI] (N2) at (0.000,8.250) {$\Nc_2$};
  \node[clone,regI] (I) at (0.000,3.750) {$\Ic$};
  \node[clone,regI] (I0) at (1.200,2.250) {$\Ic_0$};
  \node[clone,regI] (I1) at (-1.200,2.250) {$\Ic_1$};
  \node[clone,regI] (I2) at (0.000,0.750) {$\Ic_2$};
  \node[clone,regIV] (Sn0) at (-8.040,22.575) {$\Sc_0^n$};
  \node[clone,regIV] (Sn02) at (-6.960,20.625) {$\Sc_{02}^n$};
  \node[clone,regIV] (Sn01) at (-5.760,20.625) {$\Sc_{01}^n$};
  \node[clone,regIV] (Sn00) at (-4.680,18.675) {$\Sc_{00}^n$};
  \node[clone,regIV] (Sn10) at (4.680,18.675) {$\Sc_{10}^n$};
  \node[clone,regIV] (Sn11) at (5.760,20.625) {$\Sc_{11}^n$};
  \node[clone,regIV] (Sn12) at (6.960,20.625) {$\Sc_{12}^n$};
  \node[clone,regIV] (Sn1) at (8.040,22.575) {$\Sc_1^n$};
  \begin{scope}[on background layer]
  \path (R2) edge[edge] (R1);
  \path (R2) edge[edge] (R0);
  \path (R1) edge[edge] (BF);
  \path (R0) edge[edge] (BF);
  \path (M2) edge[edge] (M1);
  \path (M2) edge[edge, out=120, in=240, looseness=0.8] (R2);
  \path (M2) edge[edge] (M0);
  \path (M1) edge[edge] (R1);
  \path (M1) edge[edge] (M);
  \path (M0) edge[edge] (R0);
  \path (M0) edge[edge] (M);
  \path (M) edge[edge, out=50, in=306, looseness=0.8] (BF);
  \path (S31) edge[edge] (S21);
  \path (S21) edge[edge] (R0);
  \path (S12) edge[edge] (S1);
  \path (S312) edge[edge] (S31);
  \path (S312) edge[edge] (S212);
  \path (S212) edge[edge] (S21);
  \path (S212) edge[edge] (R2);
  \path (S11) edge[edge] (S1);
  \path (S311) edge[edge] (S31);
  \path (S311) edge[edge] (S211);
  \path (S211) edge[edge] (S21);
  \path (S211) edge[edge] (M0);
  \path (S10) edge[edge] (S12);
  \path (S10) edge[edge] (S11);
  \path (S310) edge[edge] (S312);
  \path (S310) edge[edge] (S311);
  \path (S310) edge[edge] (S210);
  \path (S210) edge[edge] (S212);
  \path (S210) edge[edge] (S211);
  \path (S210) edge[edge] (M2);
  \path (S30) edge[edge] (S20);
  \path (S20) edge[edge] (R1);
  \path (S02) edge[edge] (S0);
  \path (S302) edge[edge] (S30);
  \path (S302) edge[edge] (S202);
  \path (S202) edge[edge] (S20);
  \path (S202) edge[edge] (R2);
  \path (S01) edge[edge] (S0);
  \path (S301) edge[edge] (S30);
  \path (S301) edge[edge] (S201);
  \path (S201) edge[edge] (S20);
  \path (S201) edge[edge] (M1);
  \path (S00) edge[edge] (S02);
  \path (S00) edge[edge] (S01);
  \path (S300) edge[edge] (S302);
  \path (S300) edge[edge] (S301);
  \path (S300) edge[edge] (S200);
  \path (S200) edge[edge] (S202);
  \path (S200) edge[edge] (S201);
  \path (S200) edge[edge] (M2);
  \path (D2) edge[edge] (S200);
  \path (D2) edge[edge] (S210);
  \path (D2) edge[edge] (D1);
  \path (D1) edge[edge, out=135, in=220, looseness=0.95] (R2);
  \path (D1) edge[edge] (D);
  \path (D) edge[edge, out=41, in=306, looseness=0.85] (BF);
  \path (L2) edge[edge] (L1);
  \path (L2) edge[edge] (L3);
  \path (L2) edge[edge] (L0);
  \path (L2) edge[edge, out=50, in=310, looseness=0.6] (D1);
  \path (L1) edge[edge, out=110, in=250, looseness=0.9] (R1);
  \path (L1) edge[edge] (L);
  \path (L3) edge[edge] (L);
  \path (L3) edge[edge, out=130, in=230, looseness=0.8] (D);
  \path (L0) edge[edge] (L);
  \path (L0) edge[edge, out=70, in=290, looseness=0.9] (R0);
  \path (L) edge[edge, out=55, in=306, looseness=0.8] (BF);
  \path (E2) edge[edge] (E1);
  \path (E2) edge[edge] (S10);
  \path (E2) edge[edge] (E0);
  \path (E1) edge[edge] (M1);
  \path (E1) edge[edge] (E);
  \path (E0) edge[edge] (E);
  \path (E0) edge[edge] (S11);
  \path (E) edge[edge] (M);
  \path (V2) edge[edge] (V1);
  \path (V2) edge[edge] (S00);
  \path (V2) edge[edge] (V0);
  \path (V1) edge[edge] (S01);
  \path (V1) edge[edge] (V);
  \path (V0) edge[edge] (V);
  \path (V0) edge[edge] (M0);
  \path (V) edge[edge] (M);
  \path (N2) edge[edge] (N);
  \path (N2) edge[edge, out=50, in=310, looseness=0.8] (L3);
  \path (N) edge[edge, out=125, in=230, looseness=0.8] (L);
  \path (I2) edge[edge] (I1);
  \path (I2) edge[edge] (I0);
  \path (I2) edge[edge] (E2);
  \path (I2) edge[edge] (V2);
  \path (I2) edge[edge, out=115, in=240, looseness=0.8] (N2);
  \path (I2) edge[edge, out=115, in=240, looseness=0.8] (L2);
  \path (I2) edge[edge, out=70, in=290, looseness=0.7] (D2);
  \path (I1) edge[edge] (V1);
  \path (I1) edge[edge, out=110, in=250, looseness=0.8] (L1);
  \path (I1) edge[edge] (E1);
  \path (I1) edge[edge] (I);
  \path (I0) edge[edge] (I);
  \path (I0) edge[edge, out=70, in=290, looseness=0.8] (L0);
  \path (I0) edge[edge] (V0);
  \path (I0) edge[edge] (E0);
  \path (I) edge[edge] (V);
  \path (I) edge[edge, out=62, in=312, looseness=0.8] (N);
  \path (I) edge[edge] (E);
  \path (S0) edge[displayed gap] (Sn0);
  \path (Sn0) edge[displayed gap] (S30);
  \path (S02) edge[displayed gap] (Sn02);
  \path (Sn02) edge[displayed gap] (S302);
  \path (S01) edge[displayed gap] (Sn01);
  \path (Sn01) edge[displayed gap] (S301);
  \path (S00) edge[displayed gap] (Sn00);
  \path (Sn00) edge[displayed gap] (S300);
  \path (S10) edge[displayed gap] (Sn10);
  \path (Sn10) edge[displayed gap] (S310);
  \path (S11) edge[displayed gap] (Sn11);
  \path (Sn11) edge[displayed gap] (S311);
  \path (S12) edge[displayed gap] (Sn12);
  \path (Sn12) edge[displayed gap] (S312);
  \path (S1) edge[displayed gap] (Sn1);
  \path (Sn1) edge[displayed gap] (S31);
  \path (Sn02) edge[edge] (Sn0);
  \path (Sn01) edge[edge] (Sn0);
  \path (Sn00) edge[edge] (Sn02);
  \path (Sn00) edge[edge] (Sn01);
  \path (Sn12) edge[edge] (Sn1);
  \path (Sn11) edge[edge] (Sn1);
  \path (Sn10) edge[edge] (Sn12);
  \path (Sn10) edge[edge] (Sn11);
  \end{scope}
  \node[swatch,regI]   at (-8.3,9.6) {}; \node[anchor=west,font=\small] at (-7.9,9.6) {region~I};
  \node[swatch,regII]  at (-8.3,8.3) {}; \node[anchor=west,font=\small] at (-7.9,8.3) {region~II};
  \node[swatch,regIII] at (-8.3,7.0) {}; \node[anchor=west,font=\small] at (-7.9,7.0) {region~III};
  \node[swatch,regIV]  at (-8.3,5.7) {}; \node[anchor=west,font=\small] at (-7.9,5.7) {region~IV};
  \node[swatch,regV]   at (-8.3,4.4) {}; \node[anchor=west,font=\small] at (-7.9,4.4) {region~V};
\end{tikzpicture}

\smallskip
{\small
\begin{tabular}{@{}cllllll@{}}
\toprule
& Clones & ERM & VC dim. & PAC & Occam & SQ \\
\midrule
I   & $\Nc_*$, $\Ic_*$
    & $\Ptime$ & linear & yes & yes & yes \\
II  & $\Ec_*$, $\Vc_*$
    & no weak approximator & linear & yes & yes & yes \\
III & $\Lc_*$
    & no weak approximator & linear & yes & yes & no \\
IV  & $\Sc_{0*}^k$, $\Sc_{1*}^k$ $(k\ge2)$, $\Dc_2$
    & $\NP$-hard, $k$-approximable & exponential & no & none & no \\
V   & $\BF$, $\Rc_*$, $\Mc_*$, $\Sc_{0*}$, $\Sc_{1*}$, $\Dc$, $\Dc_1$
    & $\Ptime$ & exponential & no & none & no \\
\bottomrule
\end{tabular}}

\caption{Post's lattice, coloured by the regions the classifications cut it
into, with the status of each problem per region.  Fitting is omitted from the
table because it is in $\Ptime$ for every basis, with a fitting formula of
polynomial size (Theorem~\ref{thm:boolean-fitting-constructive}).  The Occam algorithms of
regions~I and~II are attribute-efficient, that of region~III is not known to be.  Under random
classification noise the last column is unchanged except in region~III, where
the question is open (Section~\ref{sec:noise}).  In region~IV the
$k$-approximation is optimal under the Unique Games Conjecture.  The negative
entries in the last three columns hold under the cryptographic assumptions of
Section~\ref{sec:pac}, except that in region~III the entry of the
last column is unconditional.  Dashed edges abbreviate the infinite
ascending families $\Sc_{0*}^{k}$ and $\Sc_{1*}^{k}$.}
\label{fig:post-lattice}
\end{figure}

The above results apply to fragments that are generated by sets of connectives, and that are for that
reason closed under substitution.  This excludes fragments such as Horn CNF, obtained by
restricting the shape of clauses rather than the supply of connectives.
Section~\ref{sec:constraints} recalls Boolean constraint languages as the
standard way of defining such fragments and records what becomes of the
questions above when fragments are defined this way.

\subsection*{Organization}

Section~\ref{sec:prelim} fixes notation, including the two size measures and the
table of clones.  Section~\ref{sec:fitting} treats fitting and
Section~\ref{sec:erm-regions} empirical risk minimization.
Section~\ref{sec:vc} then treats the VC dimension, Section~\ref{sec:pac} PAC
learning, Section~\ref{sec:occam} Occam algorithms, and Section~\ref{sec:noise}
learning under random classification noise and from statistical queries.
Section~\ref{sec:constraints} compares this picture with the one obtained when
fragments are given by constraint languages instead.

\subsection*{Acknowledgements}

I am grateful to Victor Dalmau and Peter Mayr for fruitful discussions.
Claude Fable was used in the process of writing this paper, both 
editorially and as a tool in the creative process. Several of the 
technical results were obtained with its help. 

\section{Preliminaries: Formulas, Circuits, Clones}\label{sec:prelim}

A \emph{basis} is a finite set $O$ of Boolean functions, each of fixed arity.
A \emph{$\PL_O$-formula} over the variables $x_1,\ldots,x_n$ is a finite tree
whose leaves are labeled by variables and whose internal nodes of arity $r$ are
labeled by $r$-ary members of $O$; we write $\PL_O$ for the class of these
formulas, and $[\varphi]$ for the Boolean function computed by $\varphi$.
A \emph{$\CIR_O$-circuit} is the same object with fan-out
allowed to exceed one, so that a subcomputation may be shared; we write
$\CIR_O$ for this class.

A \emph{clone} is a set of Boolean functions that contains all projections
$(x_1,\ldots,x_n)\mapsto x_i$ and is closed under composition.  For a set $O$
of Boolean functions, $[O]$ denotes the clone \emph{generated} by $O$, the
least clone containing $O$; it consists of the functions computed by
$\PL_O$-formulas, or equivalently by $\CIR_O$-circuits.  A finite set $O$
with $[O]=\Cc$ is a \emph{basis} of the clone $\Cc$.  Post~\cite{Post}
determined all Boolean clones: there are countably many, each has a finite
basis, and ordered by inclusion they form \emph{Post's lattice}, drawn in
Figure~\ref{fig:post-lattice}.  For finite bases $O,O'$, write
\[
  O\preceq O'
  \quad\Longleftrightarrow\quad
  [O]\subseteq[O'],
\]
and write $O\equiv O'$ when the two generated clones are equal.

Order $\{0,1\}^n$ coordinatewise, so that $\mathbf a\le\mathbf a'$ means
$a_i\le a'_i$ for every $i$.  Write $\zero$ and $\one$ for the all-zero and
the all-one tuple, and $\bar{\mathbf a}$ for the bitwise complement of
$\mathbf a$, that is, $\bar a_i=1-a_i$.  The \emph{upward closure} of a set
$A\subseteq\{0,1\}^n$ consists of the points above some member of $A$; a set
equal to its upward closure is an \emph{upset}, and \emph{downward closure}
and \emph{downset} are defined dually.  We write $\oplus^3$ for
the ternary parity operation $\oplus^3(x,y,z)=x\oplus y\oplus z$, and $f^{d}$
for the \emph{dual} of a function $f$, defined by
$f^{d}(\mathbf a)=\neg f(\bar{\mathbf a})$, and a function equal to its own dual
is \emph{self-dual}.  We use the standard notation for Post's lattice, as
in~\cite{BCRV,Lau}.

\begin{table}[tp]
\centering
{\small
\begin{tabular}{lll}
\toprule
Clone & Description & Representative basis \\
\midrule
$\BF$ & all Boolean functions & $\{\wedge,\neg\}$ \\
$\Rc_0$ & $0$-preserving, i.e.\ $f(\zero)=0$ & $\{\wedge,\oplus\}$ \\
$\Rc_1$ & $1$-preserving, i.e.\ $f(\one)=1$ & $\{\to,\wedge\}$ \\
$\Rc_2$ & $\Rc_0\cap\Rc_1$ & $\{\vee,\ x\wedge(y\leftrightarrow z)\}$ \\
\addlinespace
$\Mc$ & monotone & $\{\wedge,\vee,\bot,\top\}$ \\
$\Mc_0$ & $\Mc\cap \Rc_0$ & $\{\wedge,\vee,\bot\}$ \\
$\Mc_1$ & $\Mc\cap \Rc_1$ & $\{\wedge,\vee,\top\}$ \\
$\Mc_2$ & $\Mc\cap \Rc_2$ & $\{\wedge,\vee\}$ \\
\addlinespace
$\Sc_0$ & $0$-separating & $\{\to\}$ \\
$\Sc_{02}$ & $\Sc_0\cap \Rc_2$ & $\{x\vee(y\wedge\neg z)\}$ \\
$\Sc_{01}$ & $\Sc_0\cap \Mc$ & $\{x\vee(y\wedge z),\ \top\}$ \\
$\Sc_{00}$ & $\Sc_0\cap \Rc_2\cap \Mc$ & $\{x\vee(y\wedge z)\}$ \\
$\Sc_0^k$ & $0$-separating of degree $k$ & $\{\to,\ \thr^{k+1}_2\}$ \\
$\Sc_{02}^k$ & $\Sc_0^k\cap \Rc_2$ & $\{x\vee(y\wedge\neg z),\ \thr^{k+1}_2\}$ \\
$\Sc_{01}^k$ & $\Sc_0^k\cap \Mc$ & $\{x\vee(y\wedge z),\ \top,\ \thr^{k+1}_2\}$ \\
$\Sc_{00}^k$ & $\Sc_0^k\cap \Rc_2\cap \Mc$ & $\{x\vee(y\wedge z),\ \thr^{k+1}_2\}$ \\
\addlinespace
$\Sc_1$ & $1$-separating & $\{x\wedge\neg y\}$ \\
$\Sc_{12}$ & $\Sc_1\cap \Rc_2$ & $\{x\wedge(y\vee\neg z)\}$ \\
$\Sc_{11}$ & $\Sc_1\cap \Mc$ & $\{x\wedge(y\vee z),\ \bot\}$ \\
$\Sc_{10}$ & $\Sc_1\cap \Rc_2\cap \Mc$ & $\{x\wedge(y\vee z)\}$ \\
$\Sc_1^k$ & $1$-separating of degree $k$ & $\{x\wedge\neg y,\ \thr^{k+1}_k\}$ \\
$\Sc_{12}^k$ & $\Sc_1^k\cap \Rc_2$ & $\{x\wedge(y\vee\neg z),\ \thr^{k+1}_k\}$ \\
$\Sc_{11}^k$ & $\Sc_1^k\cap \Mc$ & $\{x\wedge(y\vee z),\ \bot,\ \thr^{k+1}_k\}$ \\
$\Sc_{10}^k$ & $\Sc_1^k\cap \Rc_2\cap \Mc$ & $\{x\wedge(y\vee z),\ \thr^{k+1}_k\}$ \\
\addlinespace
$\Dc$ & self-dual & $\{\maj,\neg\}$ \\
$\Dc_1$ & $\Dc\cap \Rc_2$ & $\{\maj(x,y,\neg z)\}$ \\
$\Dc_2$ & $\Dc\cap \Mc$ & $\{\maj\}=\{\thr^3_2\}$ \\
\addlinespace
$\Lc$ & affine & $\{\oplus,\top\}$ \\
$\Lc_0$ & $\Lc\cap \Rc_0$ & $\{\oplus\}$ \\
$\Lc_1$ & $\Lc\cap \Rc_1$ & $\{\leftrightarrow\}$ \\
$\Lc_2$ & $\Lc\cap \Rc_2$ & $\{\oplus^3\}$ \\
$\Lc_3$ & $\Lc\cap \Dc$ & $\{x\oplus y\oplus z\oplus\top\}$ \\
\addlinespace
$\Ec$ & conjunctions and constants & $\{\wedge,\bot,\top\}$ \\
$\Ec_0$ & $\Ec\cap \Rc_0$ & $\{\wedge,\bot\}$ \\
$\Ec_1$ & $\Ec\cap \Rc_1$ & $\{\wedge,\top\}$ \\
$\Ec_2$ & $\Ec\cap \Rc_2$ & $\{\wedge\}$ \\
\addlinespace
$\Vc$ & disjunctions and constants & $\{\vee,\bot,\top\}$ \\
$\Vc_0$ & $\Vc\cap \Rc_0$ & $\{\vee,\bot\}$ \\
$\Vc_1$ & $\Vc\cap \Rc_1$ & $\{\vee,\top\}$ \\
$\Vc_2$ & $\Vc\cap \Rc_2$ & $\{\vee\}$ \\
\addlinespace
$\Nc$ & projections, negations and constants & $\{\neg,\bot,\top\}$ \\
$\Nc_2$ & projections and negations & $\{\neg\}$ \\
\addlinespace
$\Ic$ & projections and constants & $\{\bot,\top\}$ \\
$\Ic_0$ & $\Ic\cap\Rc_0$: projections and $\bot$ & $\{\bot\}$ \\
$\Ic_1$ & $\Ic\cap\Rc_1$: projections and $\top$ & $\{\top\}$ \\
$\Ic_2$ & projections only & $\varnothing$ \\
\bottomrule
\end{tabular}}
\caption{The Boolean clones, with a basis for each; $k$ ranges over the
integers $k\ge2$.  The bases follow those tabulated by B\"ohler, Creignou,
Reith and Vollmer~\cite{BCRV}, up to inessential variation, and each is one
choice among many.}
\label{tab:clones}
\end{table}

Table~\ref{tab:clones} lists the clones, with a description and a basis for
each.  The list is complete~\cite{Post,Lau}.

An asterisk abbreviates a family of clones sharing a letter: $\Rc_*$ denotes
$\Rc_0,\Rc_1,\Rc_2$, and $\Mc_*$, $\Lc_*$, $\Ec_*$, $\Vc_*$, $\Nc_*$, $\Ic_*$ and $\Dc_*$ likewise
denote the clone named by the letter together with all of its subscripted
variants, so that $\Mc_*$ is $\Mc,\Mc_0,\Mc_1,\Mc_2$ and $\Lc_*$ is $\Lc,\Lc_0,\Lc_1,\Lc_2,\Lc_3$.
On the separating side $\Sc_{0*}$ denotes $\Sc_0,\Sc_{02},\Sc_{01},\Sc_{00}$ and $\Sc_{1*}$
denotes $\Sc_1,\Sc_{12},\Sc_{11},\Sc_{10}$, with $\Sc_{0*}^k$ and $\Sc_{1*}^k$ the
corresponding degree-$k$ families.

The convention for the $\Sc$-families is the following.  A Boolean function is
$0$-separating if the set $f^{-1}(0)$ has a common zero coordinate; it is
$0$-separating of degree $k$ if every list of $k$ elements of $f^{-1}(0)$,
with repetitions allowed, has a common zero coordinate.  Equivalently, every
nonempty subset of at most $k$ distinct zero inputs has a common zero
coordinate.  The $1$-separating notions are dual.

For a formula $\varphi$ we write $\size{\varphi}$ for its number of nodes
and $\dep(\varphi)$ for its depth, with a single leaf having depth $0$.  For
a circuit $C$, $\size{C}$ is its number of gates and inputs.  Unfolding a
circuit into a formula can increase size exponentially --- indeed for
$O=\{\wedge,\vee,\top,\bot\}$, $\CIR_O$-circuits are super-polynomially more
succinct than $\PL_O$-formulas~\cite{KarchmerWigderson,RazMcKenzie} --- and
no converse blow-up is possible, so circuit size is the more generous
measure: a statement proved for it is weaker than the corresponding one for
formula size when it asserts hardness, and stronger when it asserts an upper
bound.  In certain cases circuits can nevertheless be translated to formulas
in polynomial time.  If $[O]\subseteq \Ec$, $[O]\subseteq \Vc$ or
$[O]\subseteq \Lc$, the $n$-ary functions of $[O]$ are the conjunctions, the
disjunctions or the affine functions of the variables, together with
whichever truth constants the clone contains.  Which of them a given
$\CIR_O$-circuit computes is read off from its values at $n+1$ points, and
the corresponding $\PL_O$-formula has at most $n+2$ leaves.

\section{Fitting}\label{sec:fitting}

Fix a finite set of propositional variables $x_1,\ldots,x_n$, identified
with the coordinates $[n]=\{1,\ldots,n\}$, and write
$\mathbf x=(x_1,\ldots,x_n)$ for the tuple of them.  
A \emph{labeled sample} is a finite multiset
\[
  E \subseteq_{\mathrm{multi}} \{0,1\}^n \times \{0,1\},
\]
which we also write as a list $(\mathbf a_1,b_1),\ldots,(\mathbf a_m,b_m)$
of labeled examples.  A formula $\varphi\in\PL_O$ \emph{fits} $E$ if
$\varphi(\mathbf a_j)=b_j$ for $j=1,\ldots,m$, and $E$ is \emph{realizable}
for $\PL_{O}$ if some $\varphi\in\PL_O$ fits it.  The \emph{fitting problem} for
$\PL_O$ asks, given $E$, to decide whether $E$ is realizable and, if so, to
return a fitting formula.  The same definitions apply also to $\CIR_O$. 
Note that a labeled sample is realizable for $\PL_O$ iff it is
realizable for $\CIR_O$.  Since this depends on $O$ only through $[O]$, we
also say that $E$ is \emph{realizable in the clone} $[O]$.

The fitting problem has an algebraic form that is well studied in the literature,
and we lift our terminology to it.  A \emph{finite algebra} $\mathbf A=(A;F)$
is a finite set $A$ together with a finite set $F$ of finitary operations on
$A$.  \emph{Terms} over $F$ are built from variables and the operations of $F$
in the usual way, and a term $\tau(x_1,\ldots,x_n)$ induces an $n$-ary
\emph{term operation} $\tau^{\mathbf A}$ on $A$; for
$\mathbf A=(\{0,1\};O)$ the terms are the $\PL_O$-formulas and the term
operations are the functions of $[O]$.  A \emph{labeled sample over
$\mathbf A$} consists of examples $(\mathbf a_j,b_j)\in A^n\times A$, a term
$\tau$ \emph{fits} it if $\tau^{\mathbf A}(\mathbf a_j)=b_j$ for every $j$,
and it is \emph{realizable} if some term fits it.  Read columnwise, fitting
asks whether the column of labels lies in
the subalgebra of $\mathbf A^m$ generated by the columns of the variables.
This is the \emph{subpower membership problem} for $\mathbf A$, the subject
of~\cite{Mayr,BMS}, whose complexity is open in general.  A
\emph{near-unanimity operation} of arity $d\ge3$ is one returning $x$
whenever at least $d-1$ of its arguments are $x$, and a \emph{near-unanimity
term} of $\mathbf A$ is a term inducing such an operation; on $\{0,1\}$ the
majority operation and the thresholds $\thr^{k+1}_2$ and $\thr^{k+1}_k$ are
examples of near-unanimity operations.  The Baker--Pixley
theorem~\cite{BakerPixley} solves the subpower membership problem for algebras
with a near-unanimity term.  We need a slight strengthening of this result.

\begin{proposition}[Efficient Baker--Pixley interpolation]\label{prop:nu-interpolation}
Fix a finite algebra $\mathbf A$ with a near-unanimity term $\nu$ of arity
$d\ge3$, and let $k=d-1$.
\begin{enumerate}
\item A labeled sample $E$ over $\mathbf A$ is realizable if and only if every
      $E'\subseteq E$ with $|E'|\le k$ is, which can be tested in polynomial
      time.
\item From a realizable labeled sample $E$ one can moreover compute in
      polynomial time a fitting term of depth logarithmic in $|E|$.
\end{enumerate}
\end{proposition}

\begin{proof}
Item~1 is the content of the classic Baker--Pixley theorem~\cite{BakerPixley}.
We prove both items here, with a slight refinement of the classic argument:
we use a divide-and-conquer strategy to control the size and depth of the
constructed term.

Let $E$ be a labeled sample over $\mathbf A$, consisting of the labeled
examples
\[
  (\mathbf a_1,b_1),\ldots,(\mathbf a_m,b_m)\in A^n\times A.
\]
By a \emph{window} we mean a subset $S\subseteq[m]$, and we write $E|_S$ for
the sub-sample $\{(\mathbf a_j,b_j):j\in S\}$, so that the sub-samples
$E'\subseteq E$ of item~1 are the $E|_S$ with $|S|\le k$.  If $E$ is realizable then so is
every sub-sample, by the same term; this is the easy direction of item~1.

Fix a window $S$ with $|S|\le k$.  Every term induces a labeling of the points
$\mathbf a_j$, $j\in S$, that is, an element of $A^S$, and the labelings so
induced are obtained from those induced by the variables $x_1,\ldots,x_n$ by
closing under the operations of $F$, applied pointwise.  There are at most
$|A|^k$ labelings of these points altogether, so at most $|A|^k$ of the
variables induce distinct ones, the closure is reached in a bounded number of
rounds, and every labeling in it is induced by a term of bounded size, the
bound depending only on $\mathbf A$.  Hence, in time $O(n)$ plus a constant
depending on $\mathbf A$, we can test whether a term fitting $E|_S$ exists
and, if so, compute one of bounded size, which we call $\tau_S$.  There are
$O(m^k)$ windows with $|S|\le k$, so all of this takes $O(m^k n)$ time.

Suppose now that for every window $S$ with $|S|\le k$ a fitting term $\tau_S$
exists.  We extend the definition of $\tau_S$ to every window $S\subseteq[m]$
by recursion on $|S|$, in such a way that $\tau_S$ fits $E|_S$.  If $|S|>k$
then $|S|\ge d$.  Partition $S$ into $d$ nonempty blocks $B_1,\ldots,B_d$ whose
sizes differ by at most one, and let
\[
  \tau_S:=\nu\bigl(\tau_{S\setminus B_1},\ldots,\tau_{S\setminus B_d}\bigr).
\]
Each $S\setminus B_i$ is a proper subset of $S$, so the recursion terminates,
and by induction $\tau_S$ fits $E|_S$: an index $j\in S$ lies in exactly one
block $B_i$, so the $d-1$ arguments $\tau_{S\setminus B_{i'}}$ with $i'\ne i$
take the value $b_j$ at $\mathbf a_j$, and the near-unanimity operation
$\nu^{\mathbf A}$ returns that value whatever the remaining argument computes
there.  Hence $\tau_{[m]}$ fits $E$.  In particular a sample all of whose
sub-samples of at most $k$ examples are realizable is realizable, which
completes the proof of item~1.

It remains to show that the construction runs in polynomial time and that
$\tau_{[m]}$ has depth logarithmic in $m$.  If $|S|=s>k$ then every block has
at least $\lfloor s/d\rfloor$ elements, so $|S\setminus B_i|<(1-1/d)s+1$,
which is at most $(1-1/2d)\,s$ once $s\ge2d$.  The recursion therefore has
depth $O(\log m)$, with a constant depending only on $d$.  Each level of the
recursion contributes the depth of $\nu$, and the terms at the bottom have
bounded depth, so $\tau_{[m]}$ has depth $O(\log m)$.  The recursion tree has
$d^{O(\log m)}=m^{O(1)}$ nodes, and the term built at a node has size at most
$\size{\nu}$ times one plus the sizes of its arguments, so sizes grow by at most
a factor $(d+1)\size{\nu}$ per level and every term involved, $\tau_{[m]}$
included, has size $m^{O(1)}$.  The construction therefore runs in polynomial
time.
\end{proof}

\thmfitting*

\begin{proof}[Proof of Theorem~\ref{thm:boolean-fitting-constructive}]
Let $\Cc=[O]$.  Since $O$ is fixed, we may freely replace a fixed operation of
$\Cc$ by a fixed $\PL_O$-formula defining it.  The following table places every
clone of Post's lattice in one of five cases; a clone may satisfy several, and
only coverage is needed.

\begin{center}
\begin{tabular}{cl}
\toprule
Case & Clones \\
\midrule
(1) & $\BF$, $\Rc_*$, $\Mc_*$, $\Dc_*$, and $\Sc_{0*}^k,\Sc_{1*}^k$ for every $k\ge2$\\
(2) & $\Lc_*$ \\
(3) & $\Nc_*$, $\Ic_*$ \\
(4) & $\Ec_*$, $\Vc_*$ \\
(5) & $\Sc_{0*}$, $\Sc_{1*}$ \\
\bottomrule
\end{tabular}
\end{center}

\begin{enumerate}[label=(\arabic*)]
\item $\Cc$ has a near-unanimity term.  This includes $\BF$, $\Rc_*$, $\Mc_*$,
$\Dc$, $\Dc_1$ and $\Dc_2$, all of which contain $\maj=\thr^3_2$, and the clones
$\Sc_{0*}^k$ and $\Sc_{1*}^k$, which contain $\thr^{k+1}_2$ respectively
$\thr^{k+1}_k$.  Proposition~\ref{prop:nu-interpolation}, applied to the
algebra $(\{0,1\};O)$, decides fitting and in the positive case returns
a fitting $\PL_O$-formula of polynomial size and logarithmic depth.

\item $\Cc\subseteq\Lc$.  The $n$-ary functions of $\Cc$ are the affine
functions $c\oplus\bigoplus_{j\in J}x_j$ whose constant $c$ and support $J$
satisfy the at most two linear conditions over $\Ftwo$ that single out
$\Cc$ among $\Lc$, $\Lc_0$, $\Lc_1$, $\Lc_2$ and $\Lc_3$: that $c=0$, that
$c\oplus|J|\equiv1$, or that $|J|$ be odd.  Fitting is therefore a system of
linear equations over $\Ftwo$ in the unknowns $c$ and the indicator
vector of $J$, with one equation per example and those conditions, which
Gaussian elimination solves in polynomial time.  A solution is written as a
chain of connectives of $O$ over the variables of $J$, a $\PL_O$-formula with
at most $n+2$ leaves.

\item $O\preceq\{\neg,\top,\bot\}$.  This is the essentially unary case.
One tries the finitely many available unary functions on each variable, together
with the available truth constants.

\item $O\preceq\{\wedge,\top,\bot\}$ or $O\preceq\{\vee,\top,\bot\}$.  This is
the semilattice case.  For conjunctions the standard greedy construction
applies: keep exactly those variables whose columns are compatible with all
positive labels and take their conjunction, with the appropriate
truth-constant side conditions for proper subclones.  Disjunctions are dual.

\item It remains to handle the ordinary separating clones.  We give the
argument for the $1$-separating side; the $0$-separating side follows by
duality.  For $\Cc\in \Sc_{1*}$, let $\Kc_{\Cc}$ be given by
\[
\begin{array}{c|cccc}
 \Cc&\Sc_1&\Sc_{12}&\Sc_{11}&\Sc_{10}\\
 \hline
 \Kc_{\Cc}&\BF&\Rc_1&\Mc&\Mc\cap \Rc_1.
\end{array}
\]
Write $\mathbf g_i\in\{0,1\}^m$ for the column of values of the variable
$x_i$ on the sample and $\mathbf b=(b_1,\ldots,b_m)$ for the column of labels.
Every $f\in \Cc$ is bounded above by one of its inputs.  Hence a necessary
condition for fitting is that some variable column $\mathbf g_i$ dominates the
label column $\mathbf b$.  If no such column exists, reject.  Otherwise, choose
any such $i$, discard the rows on which $\mathbf g_i=\zero$, and set $x_i=1$ in
the remaining
rows.  If the original sample has a $\Cc$-fit $f$, then $f|_{x_i=1}$ belongs to
$\Kc_{\Cc}$ and fits this residual sample.  Conversely, if $u\in \Kc_{\Cc}$ fits the
residual sample, then $x_i\wedge u$ belongs to $\Cc$ and fits the original
sample.  All four clones $\Kc_{\Cc}$ contain the majority operation, so item~(1)
decides the residual problem and constructs a residual formula of polynomial
size and logarithmic depth.

To translate this formula back to the original basis, fix a basis $B_\Cc$ of
$\Kc_{\Cc}$ and, for each $\gamma\in B_\Cc$, use the guarded gate
$\widehat\gamma(x,\mathbf z)=x\wedge \gamma(\mathbf z)$.  This operation belongs to
$\Cc$ and therefore has a fixed $\PL_O$-formula.  Replacing each residual
connective by this fixed formula, and each variable $y$ by a fixed
$\PL_O$-formula for $x_i\wedge y$, maintains the invariant that the translated
node computes $x_i\wedge u$.  Each substitution multiplies size by at most a
constant per level, and the depth is logarithmic, so the result is a
$\PL_O$-formula of polynomial size.

For the canonical basis $f_{\wedge\vee}(x,y,z)=x\wedge(y\vee z)$ of $\Sc_{10}$,
the gadgets for a guarded leaf, disjunction, and conjunction are respectively
$f_{\wedge\vee}(x_i,y,y)$, $f_{\wedge\vee}(x_i,u,v)$ and
$f_{\wedge\vee}(u,v,v)$, where $u,v$ are already guarded.  For
$f_{\wedge\to}(x,y,z)=x\wedge(y\to z)$, a basis of $\Sc_{12}$, the corresponding
gadgets are $f_{\wedge\to}(x_i,x_i,y)$, $f_{\wedge\to}(x_i,u,v)$ for implication
and $f_{\wedge\to}(u,u,v)$ for conjunction.

The bases of the four residual clones also contain truth constants, and these
need guarded gates of their own.  We have $\Kc_{\Sc_{10}}=\Mc\cap \Rc_1=\Mc_1$, with
basis $\{\vee,\wedge,\top\}$, and $\Kc_{\Sc_{11}}=\Mc=[\{\vee,\wedge,\bot,\top\}]$.  The
guarded gates for the constants are
\[
  \widehat\top(x)=x\wedge\top=x,
  \qquad
  \widehat\bot(x)=x\wedge\bot=\bot,
\]
both of which lie in $\Cc$ and satisfy the invariant.  For $\Kc_{\Sc_{12}}=\Rc_1$ the two
gadgets above suffice, $\{\to,\wedge\}$ being the basis of $\Rc_1$ in
Table~\ref{tab:clones}.  For $\Kc_{\Sc_1}=\BF$ take the basis $\{\wedge,\neg\}$, whose guarded gates are
$x\wedge(y\wedge z)$ and $x\wedge\neg y$, both of which lie in $\Sc_1$.  The dual
construction handles the ordinary $0$-separating clones.
\end{enumerate}

With dual cases folded in, the table above shows these cases to be
exhaustive, giving a polynomial-time construction of a fitting
$\PL_O$-formula of polynomial size whenever fitting is possible.  A formula is
a circuit of the same size, so the statement for $\CIR_O$ follows.
\end{proof}

\section{Empirical risk minimization}\label{sec:erm-regions}

The \emph{empirical risk} of a formula $\varphi$ on a labeled sample $E$
is the fraction of misclassified examples,
\[
  \operatorname{err}_E(\varphi)=\frac{\left|\{(\mathbf a,b)\in E : \varphi(\mathbf a)\ne b\}\right|}{|E|},
\]
and the empirical risk minimization problem for $\PL_O$, ERM for $\PL_O$ for
short, asks for a formula attaining the optimal value:
given $E$, return a $\varphi\in\PL_O$ with
\[
  \operatorname{err}_E(\varphi)=\opt_O(E)
  :=\min_{\psi\in \PL_O}\operatorname{err}_E(\psi).
\]
The sample size is fixed within an instance, so minimizing this fraction is
the same as minimizing the number of mistakes, and the arguments below count
mistakes where that is more convenient.  The associated decision problem
asks, given $E$ and $t$, whether there exists $\varphi\in \PL_O$ with at most
$t$ mistakes.  Since the hypothesis is part of the output, the problem
depends on how that hypothesis may be written, and we distinguish two
versions: ERM for $\PL_O$, where the algorithm returns a $\PL_O$-formula, and
ERM for $\CIR_O$, where it may return a $\CIR_O$-circuit.  The optimum value
$\opt_O(E)$ is the same in both, since $\PL_O$ and $\CIR_O$ define the same
functions; only the output differs.  The same distinction applies to the
weak approximators of the introduction and to the approximation algorithms
below.

This section proves the classification from the introduction, which we restate.

\thmerm*

For circuits it is clear that the problem depends on $O$ only through the
clone $[O]$: if $[O]=[O']$ then every member of $O'$ is computed by some
fixed $\CIR_O$-circuit, so replacing each gate turns a $\CIR_{O'}$-circuit
into a $\CIR_O$-circuit of size larger by at most a constant factor, and
symmetrically.  For formulas no such translation is known in general, so
the basis, and not merely the clone, could in principle matter.  Nevertheless
the theorem shows that the complexity of ERM for $\PL_O$ too depends only on
$[O]$.  The explanation lies in Theorem~\ref{thm:boolean-fitting-constructive},
which shifts the algorithmic content of the problem from formulas to
relabelings of the sample.

A \emph{relabeling} of a labeled sample $E$ assigns a label to every point of
$\{0,1\}^n$ that occurs in $E$.  It is \emph{realizable} in a clone $\Cc$ if
some function of $\Cc$ takes the assigned value at every sample point.  Its
\emph{mistakes} on $E$ are the examples $(\mathbf a,b)\in E$ whose label $b$
differs from the label assigned to $\mathbf a$, and its empirical risk is the
fraction of mistakes, as for formulas.  Every formula in $\PL_O$ or circuit in
$\CIR_O$ induces, by evaluation at the sample points, a relabeling realizable
in $[O]$ with the same mistakes.  Conversely, a relabeling realizable in
$[O]$, read as a labeled sample, is realizable for $\PL_O$, so
Theorem~\ref{thm:boolean-fitting-constructive} computes from it in polynomial
time a $\PL_O$-formula, or a $\CIR_O$-circuit, with exactly those mistakes.
Hence ERM for $\PL_O$ and ERM for $\CIR_O$ are both equivalent, in polynomial
time and with every approximation ratio preserved, to what we call \emph{ERM
for the clone} $\Cc=[O]$: given $E$, find a relabeling of $E$ that is
realizable in $\Cc$ and has the fewest mistakes.  We write $\opt_\Cc(E)$ for
the least empirical risk of such a relabeling, so that
$\opt_\Cc(E)=\opt_O(E)$ for every basis $O$ of $\Cc$, and the notions of
$\alpha$-approximation and of weak approximator carry over verbatim, with a
realizable relabeling in place of the formula returned.

The rest of the section treats ERM for clones: the algorithms return
realizable relabelings, and the hardness results bound the mistakes of every
function of the clone on the samples they construct.  No basis appears
anywhere, which is why the classification depends on $O$ only through $[O]$
and holds for formulas and circuits alike, as the final clause of the theorem
asserts.

We prove the theorem region by region, keeping algorithms, hardness results
and approximation guarantees for a region together; the standard enumeration
of Post's lattice shows that the families considered below are
exhaustive~\cite{Post,Lau}.  In clone names, the interval in regime~(ii) is
$\thr^{k+1}_2\in\Cc\subseteq\Sc_0^k$, or dually
$\thr^{k+1}_k\in\Cc\subseteq\Sc_1^k$, using the bases for $\Sc_0^k$ and
$\Sc_1^k$ from Table~\ref{tab:clones}; by inspection of Post's
lattice~\cite{Post,Lau}, the clones in these intervals are exactly the eight
degree-$k$ separating clones $\Sc_{0*}^k$ and $\Sc_{1*}^k$, together with the
majority clone $\Dc_2=[\maj]=[\thr^3_2]$ when $k=2$.

Two conventions are used throughout.  First, we freely use weighted
examples, that is, repeated ones.  Given a labeled sample $E$, for each point
$\mathbf a\in\{0,1\}^n$ let
\[
  m_+(\mathbf a)=\#\{(\mathbf a,1)\in E\},\qquad
  m_-(\mathbf a)=\#\{(\mathbf a,0)\in E\}.
\]
A relabeling that assigns $1$ to $\mathbf a$ makes $m_-(\mathbf a)$ mistakes
there, and one that assigns $0$ makes $m_+(\mathbf a)$, so ERM for a clone is
a weighted labeling problem on the finite set of sample points.  Conversely,
weights can be imposed by repetition: if $W>|E|$ and we add $W$ copies of
$(\mathbf a,b)$, then every optimal relabeling of the enlarged instance
assigns $b$ to $\mathbf a$, provided some relabeling realizable in the clone
satisfies all the imposed labels.  Weight $W$ always abbreviates $W$ repeated
copies, and this large-penalty forcing is how boundary conditions are
imposed below.  Second,
an algorithm is an $\alpha$-approximation for ERM for $\Cc$ if it returns a
relabeling realizable in $\Cc$ of empirical risk at most
$\alpha\cdot\opt_\Cc(E)$.  This multiplicative guarantee differs both from an
additive one, $\opt_\Cc(E)+\varepsilon$, and from weak approximation.  The
trivial baseline behind weak approximation presumes the truth constants, which
clones such as $\Ec_2=[\wedge]$ and $\Lc_2=[\oplus^3]$ lack.  This weakens
nothing, since the constant-adjunction reduction below supplies them without
changing the sample size or the empirical risk.

We begin with that reduction, which lets the truth constants be assumed
available.  For a clone $\Cc$, let $\Cc^{\pm}$ be the clone generated by
$\Cc$ together with $\top$ and $\bot$.

\begin{lemma}\label{lem:constant-adjunction}
For every clone $\Cc$, ERM for $\Cc^{\pm}$ reduces in polynomial time to ERM
for $\Cc$.  More precisely, from a sample $E$ one constructs in polynomial
time a sample $E'$ with $|E'|=|E|$ whose relabelings realizable in $\Cc$
correspond bijectively, with the same mistakes, to the relabelings of $E$
realizable in $\Cc^{\pm}$.
\end{lemma}

\begin{proof}
Let $E$ be a sample over the variables $x_1,\ldots,x_n$.  Introduce two fresh
variables $z_0,z_1$, and replace every labeled example $(\mathbf a,b)$ by the
extended example $(\mathbf a',b)$ with
\[
  \mathbf a'(x_i)=\mathbf a(x_i),\qquad \mathbf a'(z_0)=0,
  \qquad \mathbf a'(z_1)=1.
\]
Call the transformed sample $E'$.  The map $\mathbf a\mapsto\mathbf a'$ is a
bijection between the points of $E$ and those of $E'$, so relabelings of $E$
correspond to relabelings of $E'$ with the same mistakes, and it remains to see
that this correspondence respects realizability.  Every function of
$\Cc^{\pm}$ is obtained by composing functions of $\Cc$, projections and the
two constants.  Replacing each occurrence of $\bot$ and $\top$ by $z_0$ and
$z_1$ turns an $n$-ary such function $f$ into an $(n+2)$-ary function
$g\in\Cc$ with $f(\mathbf x)=g(\mathbf x,0,1)$, and conversely every function
of this form lies in $\Cc^{\pm}$.  Since $g(\mathbf a')=g(\mathbf a,0,1)$, a
relabeling of $E$ is realized in $\Cc^{\pm}$ by $g(\mathbf x,0,1)$ if and only
if the corresponding relabeling of $E'$ is realized in $\Cc$ by $g$.
\end{proof}

Thus, adjoining constants cannot make ERM harder.  On the other hand,
adjoining constants can make ERM easier: as we will see, ERM for $\Dc_2$ is
$\NP$-hard, while $\Dc_2^{\pm}=\Mc$ is tractable by
Proposition~\ref{prop:erm-base} --- so the lemma transfers hardness only
downward, from $\Cc^{\pm}$ to $\Cc$.

\subsection{The inapproximable cases}\label{sec:weak-agnostic}

Each of the two results below concerns the clone at the top of one of the
three intervals of regime~(i), which is where its source leaves it; the
passage to the smaller clones of an interval is uniform and is deferred to
Section~\ref{sec:erm-assembly}.

\begin{theorem}[H\aa stad~\cite{Hastad}]
\label{thm:weak-agnostic}
Unless $\Ptime=\NP$, there is no polynomial-time weak approximator for the
clone $\Lc=[\oplus,\top,\bot]$.
\end{theorem}

The functions of $\Lc$ are those of the form
$h(\mathbf x)=c\oplus\bigoplus_{j\in J}x_j$, and the theorem is H\aa stad's gap
hardness for Max-E3Lin read through the standard translation between linear
systems and parity examples.  His theorem says that, for every $\eta>0$, it is
$\NP$-hard to distinguish systems of three-variable linear equations over
$\Ftwo$ for which some assignment satisfies at least a $1-\eta$ fraction of the
equations from systems in which every assignment satisfies at most a
$1/2+\eta$ fraction~\cite{Hastad}.  Since the equations have odd arity,
complementing all variables turns satisfied equations into unsatisfied ones, so
in the second case every assignment also satisfies at least a $1/2-\eta$
fraction.  Translate such a system into labeled parity examples in the standard
way: an equation $\mathbf a\cdot\mathbf z=b$ becomes the example with coordinate
vector $\mathbf a$ and label $b$.  An assignment $\mathbf z$ is then the linear
parity $h(\mathbf x)=\mathbf x\cdot\mathbf z$, whose value on that example is
$\mathbf a\cdot\mathbf z$, and an optional affine offset corresponds to
complementing all predictions.  The translation therefore gives samples for
which either some affine hypothesis has empirical risk at most $\eta$, or every
affine hypothesis has empirical risk at least $1/2-\eta$.  Given
$\delta,\varepsilon>0$, choose $\eta<\min\{\delta,\varepsilon\}$: a weak
approximator with these constants would separate the two cases.

\begin{theorem}[Feldman, Gopalan, Khot and Ponnuswami~\cite{FGKP}]
\label{thm:weak-agnostic-monomials}
Unless $\Ptime=\NP$, there is no polynomial-time weak approximator for the
clone $\Ec=[\wedge,\top,\bot]$, nor for the clone $\Vc=[\vee,\top,\bot]$.
\end{theorem}

For $\Ec=[\wedge,\top,\bot]$ this is the monotone case of the agnostic-learning
hardness of~\cite{FGKP}, whose proof does not pass through H\aa stad's but goes
back to the PCP theorem through Feige's multi-prover proof system for
$3$SAT-$5$, precisely in order to avoid an intermediate optimization problem.
One point of care: their monomials are conjunctions of \emph{literals}, whereas
$\Ec$ contains only monotone conjunctions and the two truth constants.  The
result we need is therefore the monotone version, which they prove alongside the
general one --- it is their problem $\mathrm{MMon}\text{-}\mathrm{MA}$, treated
in~\cite[\S4.2.2 and Theorem~15]{FGKP}.  In the terminology used here it says
that, unless $\Ptime=\NP$, no polynomial-time algorithm is a weak approximator
for the class of monotone conjunctions.   The clone $\Vc=[\vee,\top,\bot]$
satisfies the same statement by Boolean duality: complement all coordinates and
flip all labels.

\subsection{Finite-degree separating clones and the majority clone}

Regime~(ii) consists of the finite-degree separating clones and, as the case
$k=2$, the majority clone.  Fix $k\ge 2$.  The proofs treat the
$0$-separating side, and the $1$-separating side follows by duality: each
clone $\Sc_{1*}^k$ consists of the duals $f^d$ of the functions $f$ of the
clone $\Sc_{0*}^k$ with the same subscript, and $f^d$ makes on the sample
$E^d=\{(\bar{\mathbf a},1-b):(\mathbf a,b)\in E\}$ the same number of
mistakes as $f$ on $E$, so ERM for a $1$-separating clone is ERM for its
dual clone with all coordinates complemented and all labels flipped.  By the
\emph{zero set} of a function $f\colon\{0,1\}^n\to\{0,1\}$ we mean the set
$f^{-1}(0)$.

\begin{fact}[Zero-set criterion]\label{fact:zero-set}
For $Z\subseteq\{0,1\}^n$, some function of $\Sc_0^k$ has zero set exactly $Z$ if
and only if every nonempty subset of $Z$ of size at most $k$ has a common zero
coordinate.
\end{fact}

This is a restatement of the separating condition of
Section~\ref{sec:prelim}.

The algorithm is obtained by the technique of Hochbaum~\cite{Hochbaum} for
weighted vertex cover: write the problem as an integer linear program, solve
its linear relaxation, and round every variable of value at least $1/k$ up
to $1$ and every other variable down to $0$.  The following lemma will justify
the rounding step. It is stated separately as it is used in two proofs.

\begin{lemma}[Rounding]\label{lem:rounding}
Let $k\ge2$ be fixed.  Consider an integer program with variables
$x_v\in\{0,1\}$ for $v\in V$, whose objective is to minimize
\[
  \sum_{v\in V}\bigl(c_v\,x_v+d_v\,(1-x_v)\bigr)
\]
with coefficients $c_v,d_v\ge0$, and whose constraints are of three kinds:
\emph{covering constraints} $\sum_{v\in T}x_v\ge1$ for sets $T\subseteq V$
with $|T|\le k$, \emph{precedence constraints} $x_v\le x_w$, and
\emph{fixed variables} $x_v=0$ or $x_v=1$.  If the program is feasible, then
a feasible solution of cost at most $k$ times the optimum can be found in
polynomial time.
\end{lemma}

\begin{proof}
Relax the integrality requirement to $x_v\in[0,1]$, solve the resulting
linear program in polynomial time, and let $\mathbf x$ be an optimal
solution; its cost is at most the optimum of the integer program.  Round
$\mathbf x$ to the $0$/$1$-assignment $\hat{\mathbf x}$ with $\hat x_v=1$ if
$x_v\ge1/k$ and $\hat x_v=0$ otherwise.  The rounded assignment satisfies all
constraints.  In a covering constraint, at most $k$ variables sum to at least
$1$, so one of them has value at least $1/k$ and is rounded to $1$.  If
$x_v\le x_w$ then $\hat x_v\le\hat x_w$.  And a fixed variable is unchanged by
rounding.

For the cost, compare termwise.  If $\hat x_v=1$ then $x_v\ge1/k$, so
$\hat x_v\le k\,x_v$; and if $\hat x_v=0$ then $x_v<1/k$, so
$1-\hat x_v=1\le\tfrac{k}{k-1}(1-x_v)$.  Each inequality is trivial in the
other case.  Hence the cost of $\hat{\mathbf x}$ is at most
$\max\{k,\tfrac{k}{k-1}\}$ times the cost of $\mathbf x$, and that maximum
is $k$ for every $k\ge2$.
\end{proof}

\begin{proposition}[$k$-approximation]\label{prop:k-approx}
Let $k\ge2$ and let $\Cc$ be one of the eight degree-$k$ separating clones
$\Sc_{0*}^k$ and $\Sc_{1*}^k$, that is, $\Sc_{00}^k\subseteq\Cc\subseteq\Sc_0^k$
or $\Sc_{10}^k\subseteq\Cc\subseteq\Sc_1^k$.  Then ERM for $\Cc$ admits a
polynomial-time $k$-approximation on arbitrary samples.
\end{proposition}

\begin{proof}
By duality we may assume $\Sc_{00}^k\subseteq\Cc\subseteq\Sc_0^k$.  We write
ERM for $\Cc$ as an integer program of the form in
Lemma~\ref{lem:rounding}.  For each distinct sample point $\mathbf a$
introduce a variable $x_{\mathbf a}\in\{0,1\}$, with $x_{\mathbf a}=1$
meaning that $\mathbf a$ is relabeled $1$, so that the number of mistakes of
the relabeling is
\[
  \sum_{\mathbf a}\Bigl(x_{\mathbf a}\,m_-(\mathbf a)+(1-x_{\mathbf a})\,m_+(\mathbf a)\Bigr).
\]
Impose a covering constraint $\sum_{\mathbf a\in T}x_{\mathbf a}\ge1$ for
every nonempty set $T$ of at most $k$ sample points without a common zero
coordinate; when $\Cc\subseteq\Mc$, a precedence constraint
$x_{\mathbf a}\le x_{\mathbf a'}$ for all sample points
$\mathbf a\le\mathbf a'$; and when $\Cc\subseteq\Rc_2$, the fixed variable
$x_{\zero}=0$, if $\zero$ is a sample point.  Since $k$ is fixed, the program
has polynomial size.  Note that $x_{\one}=1$ is among the covering
constraints whenever $\one$ is a sample point, as $\one$ has no zero
coordinate.

The $0$/$1$-assignments satisfying these constraints are exactly the
relabelings of $E$ realizable in $\Cc$.  If a relabeling is realized by
$f\in\Cc$, then its points relabeled $0$ lie in the zero set of $f$, which
satisfies the degree-$k$ condition by the zero-set criterion, so no set $T$
as above consists of points relabeled $0$, and the covering constraints
hold; when $\Cc\subseteq\Mc$, the zero set of the monotone $f$ is a downset,
which gives the precedence constraints; and when $\Cc\subseteq\Rc_2$,
$f(\zero)=0$ gives $x_{\zero}=0$.  Conversely, let $\mathbf x$ be a
$0$/$1$-assignment satisfying the constraints, and let $Z$ be the set of
sample points with $x_{\mathbf a}=0$.  Let $Z'$ be $Z$ itself, or its
downward closure in $\{0,1\}^n$ when $\Cc\subseteq\Mc$, in either case with
$\zero$ added when $\Cc\subseteq\Rc_2$.  By the covering constraints every
nonempty subset of $Z$ of size at most $k$ has a common zero coordinate, and
this passes to $Z'$, since a coordinate that is zero on a point is zero on
every point below it, and since $\zero$ is zero in every coordinate; so by
Fact~\ref{fact:zero-set}
some function $f$ of $\Sc_0^k$ has zero set exactly $Z'$.  When
$\Cc\subseteq\Mc$ this $f$ is monotone, its zero set being a downset, and
when $\Cc\subseteq\Rc_2$ it satisfies $f(\zero)=0$ and $f(\one)=1$, as
$\zero\in Z'$ while $\one$, having no zero coordinate, is neither in $Z$ nor
below a point of $Z$; so $f\in\Cc$.  Finally $f$ induces the relabeling
$\mathbf x$: it is $0$ on $Z$, and a sample point $\mathbf a$ with
$x_{\mathbf a}=1$ lies outside $Z'$, since $\mathbf a\notin Z$, since
$\mathbf a\le\mathbf z$ with $\mathbf z\in Z$ would violate the precedence
constraint $x_{\mathbf a}\le x_{\mathbf z}=0$, and since $\mathbf a=\zero$
would violate $x_{\zero}=0$.

The program is feasible, because the projection onto the first coordinate
lies in every clone and realizes some relabeling of $E$, and its optimum is
$|E|\cdot\opt_\Cc(E)$.  Lemma~\ref{lem:rounding} therefore returns in
polynomial time a $0$/$1$-assignment satisfying the constraints, that is, a
relabeling of $E$ realizable in $\Cc$, with at most
$k\cdot|E|\cdot\opt_\Cc(E)$ mistakes.
\end{proof}

We handle the majority clone $\Dc_2=[\maj]$ separately, as the approach of
Proposition~\ref{prop:k-approx} does not go through for it.  In terms of
variables $x_{\mathbf a}$ for the labels of the sample points, membership in
$\Dc_2$ imposes, besides precedence constraints, the constraints
$x_{\mathbf a}+x_{\mathbf a'}\le1$ for all sample points with
$\mathbf a'\le\bar{\mathbf a}$: if $f(\mathbf a)=f(\mathbf a')=1$ for such
points then monotonicity gives $f(\bar{\mathbf a})=1$, contradicting
self-duality.  Such packing constraints do not survive the rounding of
Lemma~\ref{lem:rounding}, which may round both variables up.  We therefore
change variables, and let a variable record whether the relabeling gives up
the majority label at a sample point, that is, the label occurring more often
there; in these variables the constraints are covering constraints.

\begin{proposition}[$2$-approximation, majority case]\label{prop:2-approx-d2}
ERM for $\Dc_2$ admits a polynomial-time $2$-approximation on arbitrary
samples.
\end{proposition}

\begin{proof}
For each distinct sample point $\mathbf a$ let its \emph{majority label}
$b_{\mathbf a}$ be $1$ if $m_+(\mathbf a)\ge m_-(\mathbf a)$ and $0$
otherwise, and let $\mu(\mathbf a)=|m_+(\mathbf a)-m_-(\mathbf a)|$.  A
relabeling makes $\min\{m_-(\mathbf a),m_+(\mathbf a)\}$ mistakes at
$\mathbf a$ if it assigns $b_{\mathbf a}$ there, and $\mu(\mathbf a)$ more
if it does not.  Since $\Dc_2\subseteq\Rc_2$, a relabeling realizable in
$\Dc_2$ assigns $0$ to $\zero$ and $1$ to $\one$, so its number of mistakes
is $B+\sum\mu(\mathbf a)$, where the sum ranges over the sample points
$\mathbf a\notin\{\zero,\one\}$ to which it does not assign $b_{\mathbf a}$,
and where
\[
  B=m_+(\zero)+m_-(\one)+\sum_{\mathbf a\notin\{\zero,\one\}}\min\{m_-(\mathbf a),m_+(\mathbf a)\}
\]
does not depend on the relabeling.

Call two sample points $\mathbf a,\mathbf a'\notin\{\zero,\one\}$
\emph{conflicting} if no function of $\Dc_2$ takes the value
$b_{\mathbf a}$ at $\mathbf a$ and $b_{\mathbf a'}$ at $\mathbf a'$, that is,
if the labeled sample $\{(\mathbf a,b_{\mathbf a}),(\mathbf a',b_{\mathbf a'})\}$
is not realizable in $\Dc_2$; this can be decided in polynomial time by
Theorem~\ref{thm:boolean-fitting-constructive}.  If $P$ is a set of sample
points outside $\{\zero,\one\}$, no two of which conflict, then the labeled
sample $\{(\mathbf a,b_{\mathbf a}):\mathbf a\in P\}$ is realizable in
$\Dc_2$.  Indeed, its sub-samples of two examples are realizable by
assumption, and those of one example by a projection, as a point
$\mathbf a\notin\{\zero,\one\}$ has a coordinate equal to $b_{\mathbf a}$;
since $\maj$ is a near-unanimity term of $\Dc_2$ of arity $3$,
Proposition~\ref{prop:nu-interpolation} with $k=2$ gives the claim.

Now introduce, for each distinct sample point $\mathbf a\notin\{\zero,\one\}$,
a variable $y_{\mathbf a}\in\{0,1\}$, with $y_{\mathbf a}=1$ meaning that
the majority label is given up at $\mathbf a$; impose the covering
constraint $y_{\mathbf a}+y_{\mathbf a'}\ge1$ for every conflicting pair;
and minimize $\sum_{\mathbf a}\mu(\mathbf a)\,y_{\mathbf a}$.  This is a
program of the form in Lemma~\ref{lem:rounding} with $k=2$, and it is
feasible, since setting every variable to $1$ satisfies all constraints.
Every $f\in\Dc_2$ gives a feasible assignment, namely $y_{\mathbf a}=1$
exactly when $f(\mathbf a)\ne b_{\mathbf a}$: it is feasible because $f$
witnesses that two points with $y_{\mathbf a}=y_{\mathbf a'}=0$ do not
conflict, and its cost is the number of mistakes of $f$ minus $B$.  So the
optimum of the program is at most $|E|\cdot\opt_{\Dc_2}(E)-B$.  Conversely,
let $\mathbf y$ be a feasible $0$/$1$-assignment and $P$ the set of points
with $y_{\mathbf a}=0$.  No two points of $P$ conflict, so by the claim the
labeled sample $\{(\mathbf a,b_{\mathbf a}):\mathbf a\in P\}$ is realizable
in $\Dc_2$.  Theorem~\ref{thm:boolean-fitting-constructive} fits it in
polynomial time, and evaluating the fitting formula at all sample points
gives a relabeling of $E$ realizable in $\Dc_2$ that assigns
$b_{\mathbf a}$ to every $\mathbf a\in P$, hence makes at most
$B+\sum_{\mathbf a}\mu(\mathbf a)y_{\mathbf a}$ mistakes.

The algorithm applies Lemma~\ref{lem:rounding} to the program and returns
the relabeling just described for the assignment obtained.  Its number of
mistakes is at most
\[
  B+2\bigl(|E|\cdot\opt_{\Dc_2}(E)-B\bigr)\le 2\,|E|\cdot\opt_{\Dc_2}(E).\qedhere
\]
\end{proof}

The matching lower bounds for all clones of the regime come from vertex cover.  A \emph{$k$-uniform
hypergraph} is a pair $H=(V,\mathcal E)$ consisting of a finite set $V$ of
vertices and a set $\mathcal E$ of \emph{hyperedges}, which are $k$-element
subsets of $V$; a graph is a $2$-uniform hypergraph.  A \emph{vertex cover}
of $H$ is a set of vertices that meets every hyperedge, and $\tau(H)$
denotes the least size of a vertex cover.

\begin{theorem}[Vertex cover in $k$-uniform hypergraphs]
\label{thm:vc-facts}
Let $k\ge2$ be fixed, and consider minimum vertex cover in $k$-uniform
hypergraphs.
\begin{enumerate}[label=(\alph*)]
\item It can be approximated within a factor $k$ in polynomial
time~\cite{Hochbaum}.
\item It is $\NP$-hard to approximate within $k-1-\varepsilon$ for every
$\varepsilon>0$ when $k\ge3$~\cite{DGKR}, and within $\sqrt2-\varepsilon$ when
$k=2$~\cite{KMS17,KMS23}.
\item Under the Unique Games Conjecture it is $\NP$-hard to approximate within
$k-\varepsilon$ for every $\varepsilon>0$~\cite{KhotRegev}.
\end{enumerate}
\end{theorem}

Part~(a) is the case of Lemma~\ref{lem:rounding} in which all constraints
are covering constraints; we use only the lower bounds~(b) and~(c).

\begin{proposition}[Vertex cover reduces to ERM]\label{prop:vc-to-erm}
Let $k\ge2$ and let $\Cc$ be a clone with $\thr^{k+1}_2\in\Cc\subseteq\Sc_0^k$
or $\thr^{k+1}_k\in\Cc\subseteq\Sc_1^k$, that is, one of the eight
degree-$k$ separating clones $\Sc_{0*}^k$ and $\Sc_{1*}^k$, or $\Dc_2$ when
$k=2$.  Then minimum vertex cover in $k$-uniform hypergraphs reduces to ERM
for $\Cc$ in polynomial time, by a reduction that preserves the objective
value and produces samples consisting of negative examples only, or of
positive examples only on the $1$-separating side.  The same sample serves
all clones of one side.
\end{proposition}

\begin{proof}
By duality we may assume $\thr^{k+1}_2\in\Cc\subseteq\Sc_0^k$.  Let
$H=(V,\mathcal E)$ be a $k$-uniform hypergraph.  Call a
set $X\subseteq V$ \emph{hyperedge-free}
if no hyperedge is a subset of $X$.  Since every hyperedge has exactly $k$
vertices, a set of at most $k$ vertices is hyperedge-free exactly when it is
not itself a hyperedge.  We create one point $\mathbf p_v$ for every vertex
$v\in V$.  For every hyperedge-free $X\subseteq V$ with $|X|\le k$, introduce a
coordinate $c_X$, and define
\[
  \mathbf p_v(c_X)=0
  \quad\Longleftrightarrow\quad
  v\in X.
\]
Since $k$ is fixed, the number of coordinates is polynomial in $|V|$.
The construction has the following property: for every nonempty
$Y\subseteq V$ with $|Y|\le k$,
\[
  \{\mathbf p_v:v\in Y\}\text{ has a common zero coordinate}
  \quad\Longleftrightarrow\quad
  Y\text{ is not a hyperedge of }H.
\]
Indeed, a common zero coordinate of these points is a coordinate $c_X$ with
$Y\subseteq X$.  If one exists then $Y$ is hyperedge-free, being a subset of the
hyperedge-free set $X$, and if $Y$ is hyperedge-free then $c_Y$ is such a
coordinate.  Label every $\mathbf p_v$ negatively.

The two directions of the reduction are established for different clones:
from a hypothesis in the largest clone $\Sc_0^k$ we extract a vertex cover
whose size is its number of mistakes, and from a vertex cover we build a
hypothesis with at most that many mistakes in $\Sc_{00}^k$ and, when $k=2$,
also in $\Dc_2$.  Since $\Cc\subseteq\Sc_0^k$ contains $\Sc_{00}^k$ or is
$\Dc_2$, both directions apply to $\Cc$.

\emph{From hypotheses to covers.}  Let $f\in\Sc_0^k$ misclassify the examples
indexed by $S_f\subseteq V$, so that the correctly classified examples are
indexed by $I_f=V\setminus S_f$.  Since
$\{\mathbf p_v:v\in I_f\}\subseteq f^{-1}(0)$ and $f$ is $0$-separating of
degree $k$, every $k$ points among these have a common zero coordinate.  If
some hyperedge $e\in\mathcal E$ were contained in $I_f$, then the $k$ points
$\{\mathbf p_v:v\in e\}$ would have a common zero coordinate, contradicting the
property above.  Thus $S_f$ meets every hyperedge, so $S_f$ is a vertex cover
of $H$, of size equal to the number of mistakes of $f$.  

\emph{From covers to hypotheses.}  Let $S\subseteq V$ be a vertex cover, and
put $I=V\setminus S$, which is hyperedge-free.  Let $Z$ be the downset of
$\{0,1\}^n$ generated by the points $\mathbf p_v$ with $v\in I$, together with
$\zero$:
\[
  Z=\{\mathbf q:\mathbf q\le\mathbf p_v\text{ for some }v\in I\}\cup\{\zero\},
\]
and let $f$ be the function with zero set $Z$.  We check that $f\in\Sc_{00}^k$.
First, $Z$ satisfies the degree-$k$ condition.  Any $k$ points of $Z$ are
$\zero$ or lie below points $\mathbf p_{v_1},\ldots,\mathbf p_{v_r}$ with
$v_1,\ldots,v_r\in I$ and $r\le k$.  The set $\{v_1,\ldots,v_r\}$ is a
hyperedge-free subset of $I$, so the coordinate $c_{\{v_1,\ldots,v_r\}}$ is zero
on each $\mathbf p_{v_i}$, hence on every point below one of them, and on
$\zero$.  Here we use that a coordinate that is zero on a point is zero on
every point below it, so that the degree-$k$ condition passes from a set of
points to the downset it generates.  By the zero-set criterion,
Fact~\ref{fact:zero-set}, $f\in\Sc_0^k$.  Second, $f$ is monotone, because
its zero set is a downset.  Third, $f(\zero)=0$ because $\zero\in Z$, and
$f(\one)=1$ because $\one\notin Z$: every $\mathbf p_v$ has the zero coordinate
$c_{\{v\}}$, a single vertex being hyperedge-free as $k\ge2$, and no point
below a point with a zero coordinate is $\one$.  So $f\in\Sc_0^k\cap\Mc\cap\Rc_2=\Sc_{00}^k$.
Finally $f(\mathbf p_v)=0$ for every $v\in I$, so $f$ makes at most $|S|$
mistakes.  It may make fewer, if some $\mathbf p_v$ with $v\in S$ happens to lie
in $Z$, and that only helps.  This direction produces a hypothesis in
$\Sc_{00}^k$, hence in each of the four clones $\Sc_{0*}^k$.

\emph{From covers to hypotheses in $\Dc_2$.}  Let $k=2$, so that $H$ is a
graph, and let again $S$ be a vertex cover and $I=V\setminus S$, which is
now an independent set.  We need $f\in\Dc_2$ with $f(\mathbf p_v)=0$ for
all $v\in I$, that is, the labeled sample $\{(\mathbf p_v,0):v\in I\}$ must
be realizable in $\Dc_2$.  Its sub-samples of at most two examples are
realizable by projections: for $u,v\in I$, distinct or not, the set
$\{u,v\}$ is a hyperedge-free subset of $I$, so $c_{\{u,v\}}$ is a common
zero coordinate of $\mathbf p_u$ and $\mathbf p_v$.  Since $\maj$ is a
near-unanimity term of $\Dc_2$ of arity $3$,
Proposition~\ref{prop:nu-interpolation} with $k=2$ gives $f$, and it makes
at most $|S|$ mistakes.

\emph{Conclusion.}  By the first direction every hypothesis in
$\Cc\subseteq\Sc_0^k$ makes at least $\tau(H)$ mistakes, and by the second
some hypothesis in $\Cc$ makes at most $\tau(H)$, so the optimum of the ERM
instance is exactly $\tau(H)$, for each clone $\Cc$ of the interval.
\end{proof}

An objective-preserving reduction is in particular approximation-preserving, so
parts~(b) and~(c) of Theorem~\ref{thm:vc-facts} transfer to ERM for all
clones of Proposition~\ref{prop:vc-to-erm}.  The factor $k$ is therefore
optimal under the Unique Games Conjecture.

\subsection{The tractable cases}

Every clone outside the two hard regimes admits a polynomial-time ERM
algorithm.  Two elementary algorithms do all the work --- pointwise majority for
$\BF$, minimum cut for $\Mc$ --- and the remaining cases reduce to them by forcing
labels or enumerating a coordinate.

\begin{proposition}[Base ERM algorithms]\label{prop:erm-base}
ERM for $\BF$ and ERM for $\Mc$ are solvable in polynomial time.
\end{proposition}

\begin{proof}
For $\BF$, every relabeling of the sample is realizable, and the values at
distinct sample points are independent.  Hence the relabeling that assigns to
each sample point the label occurring more often at it is optimal.

For $\Mc$, the positive region of a monotone function is an upset in the
coordinatewise order on $\{0,1\}^n$.  After merging duplicate examples, choosing
an upset $U$ has cost
\[
  \sum_{\mathbf a\in U}m_-(\mathbf a)+\sum_{\mathbf a\notin U}m_+(\mathbf a).
\]
Thus the problem is the minimum-cost upset problem on the finite poset induced
by the sample points.  Equivalently, its complement is a minimum-cost downset,
which is a standard minimum-cost closure problem~\cite{Picard} and can be solved by one
$s$--$t$ min-cut computation.  Every upset of the induced sample poset extends
to an upset of the full Boolean cube by taking its upward closure, so the
computed relabeling of the sample is realizable by a monotone Boolean
function.
\end{proof}

\begin{theorem}[The tractable regime]\label{thm:erm-tractable}
ERM for $\Cc$ is solvable in polynomial time whenever $\Cc$ is one of
$\BF$, $\Rc_*$, $\Mc_*$, $\Sc_{0*}$, $\Sc_{1*}$, $\Dc$, $\Dc_1$, $\Nc_*$ or $\Ic_*$.
\end{theorem}

\begin{proof}
The clones $\BF$ and $\Mc$ are Proposition~\ref{prop:erm-base}.  Every other case
reduces to one of these two, except the unary clones, which are solved by
enumeration.

\emph{Boundary clones.}  For $\Rc_*$, reduce to ERM for $\BF$ by boundary forcing:
add $W$ copies of $(\zero,0)$ for $\Rc_0$, of $(\one,1)$ for $\Rc_1$, and of both
for $\Rc_2$.  Every optimum of the enlarged sample then satisfies the required
boundary condition, and conversely $\Rc_i$ consists of exactly the Boolean
functions satisfying that condition, so the two instances have the same optimum
value.  The same argument reduces $\Mc_0,\Mc_1,\Mc_2$ to ERM for $\Mc$.

\emph{Ordinary separating clones.}  A function $f$ lies in $\Sc_0$ precisely
when its zero set is contained in $\{\mathbf a:a_i=0\}$ for some coordinate
$i$, that is, when
setting $x_i$ to $1$ forces the value $1$.  For a fixed $i$, add $W$ positive
copies of every sample point $\mathbf a$ with $a_i=1$, run the $\BF$
algorithm, let $h$ be a Boolean function realizing the relabeling it returns,
and put $f(\mathbf x)=x_i\vee h(\mathbf x)$.
On the sample points $\mathbf a$ with $a_i=0$ the two agree, and on the remaining
ones both predict $1$ because of the
forcing.  Moreover the zero set of $f$ is contained in
$\{\mathbf a:a_i=0\}$, so $f\in \Sc_0$ and
the relabeling is realizable in $\Sc_0$.
Trying all $i\in[n]$ and keeping the best solution puts
ERM for $\Sc_0$ in $\Ptime$.  Dually, for $\Sc_1$, force
every sample point $\mathbf a$ with $a_i=0$ to be negative and replace $h$ by
$x_i\wedge h(\mathbf x)$.

The six remaining clones combine this coordinate enumeration with the boundary
forcing.  For $\Sc_{02}$ and $\Sc_{12}$, add the $\Rc_2$ boundary copies to the
enlarged $\BF$-instance.  For the monotone ones, reduce to ERM for $\Mc$ instead:
for $\Sc_{01}=\Sc_0\cap \Mc$, fix $i$, force the points $\mathbf a$ with $a_i=1$ positive and
solve the resulting $\Mc$-instance; for $\Sc_{00}=\Sc_0\cap \Rc_2\cap \Mc$, add the
boundary labels $\zero\mapsto0$ and $\one\mapsto1$ as well; and $\Sc_{11},\Sc_{10}$
are dual.  The same extensions apply, since $x_i\vee h$ is monotone and
$0$-separating, and $x_i\wedge h$ monotone and $1$-separating, whenever $h$ is
monotone.

\emph{The self-dual clones $\Dc$ and $\Dc_1$.}  Let $\rho$ choose a canonical
representative of each complement pair, say the lexicographically smaller of
$\mathbf a$ and $\bar{\mathbf a}$, so that
$\rho(\mathbf a)\in\{\mathbf a,\bar{\mathbf a}\}$ and
$\rho(\mathbf a)=\rho(\bar{\mathbf a})$, and
let $\sigma(\mathbf a)=0$ if $\mathbf a=\rho(\mathbf a)$ and
$\sigma(\mathbf a)=1$ otherwise.  A self-dual
function is determined by its values on the representatives: writing $g$ for its
restriction to them,
\[
  f(\mathbf a)=g(\rho(\mathbf a))\oplus\sigma(\mathbf a),
\]
and $g$ ranges over all Boolean functions on the representatives as $f$ ranges
over $\Dc$.  Transform each labeled example $(\mathbf a,b)$ into
$\bigl(\rho(\mathbf a),\,b\oplus\sigma(\mathbf a)\bigr)$ and call the
resulting sample $E^\rho$.
Since $f(\mathbf a)=b$ if and only if
$g(\rho(\mathbf a))=b\oplus\sigma(\mathbf a)$, mistakes are
preserved example by example, so $\opt_\Dc(E)=\opt_{\BF}(E^\rho)$.  For
$\Dc_1=\Dc\cap \Rc_2$, use the same normalization and force the pair
$\{\zero,\one\}$ to the orientation $\zero\mapsto0$, $\one\mapsto1$.  With the
lexicographic representative this is the single forced normalized label
$(\zero,0)$.

\emph{Unary and projection clones.}  The clones $\Nc_*$ and $\Ic_*$ contain only
constants, projections and negated projections, so there are $O(n)$ candidate
hypotheses --- $\bot$, $\top$, $x_i$ and $\neg x_i$, the admissible subset depending on
the clone --- and enumerating them gives the optimum.
\end{proof}

\subsection{Putting the pieces together}\label{sec:erm-assembly}

The hardness results of regime~(i) were stated for the clones $\Ec$, $\Vc$ and
$\Lc$ at the top of their intervals; the passage from there to the rest of each
interval is the same in all three cases, and it is the only step that is not
already in the sources.  Suppose $\Cc=[O]$ is not the clone at the top of its
interval.  Adjoining the truth constants then generates that clone: $\Cc^{\pm}$
is $\Ec$, $\Vc$ or $\Lc$ respectively.  By Lemma~\ref{lem:constant-adjunction},
ERM for $\Cc^{\pm}$ reduces to ERM for $\Cc$ by a reduction that leaves the
sample size unchanged and preserves the mistakes of every relabeling, so it
preserves empirical risks and promises as well, and a weak approximator for
$\Cc$ would give one for $\Cc^{\pm}$.  Finally, an exact ERM algorithm is a
weak approximator, since on the promised instances it returns a relabeling of
empirical risk at most $\varepsilon$.  So ERM for $\Cc$ is $\NP$-hard in all of
these cases, and by the observation in the introduction it has no
constant-factor approximation either.

Theorem~\ref{thm:erm} follows.  Regime~(i) is
Theorems~\ref{thm:weak-agnostic} and~\ref{thm:weak-agnostic-monomials} together
with the propagation just described; regime~(ii) is
Propositions~\ref{prop:k-approx} and~\ref{prop:2-approx-d2} together with
Proposition~\ref{prop:vc-to-erm} applied to Theorem~\ref{thm:vc-facts}; and
regime~(iii) is
Theorem~\ref{thm:erm-tractable}, the enumeration of Post's lattice showing
that the three regimes leave no clone unaccounted for.  All of this is proved
for ERM for the clone $[O]$, which by the discussion at the beginning of the
section is equivalent, with approximation ratios preserved, to ERM for
$\PL_O$ and to ERM for $\CIR_O$, so the theorem holds in both forms.

\section{VC dimension}\label{sec:vc}

In preparation for the next section, where we study PAC learnability, 
we clarify the VC dimension of each fragment. 
A set $A\subseteq\{0,1\}^n$ is \emph{shattered} by a class $\mathcal F$ of
$n$-ary Boolean functions if every subset of $A$ is of the form
$A\cap f^{-1}(1)$ with $f\in\mathcal F$, and the \emph{VC dimension} of
$\mathcal F$ is the largest size of a shattered set.  

The VC dimension of
$\PL_O$ (or, equivalently, of $\CIR_O$) grows either linearly or exponentially
in the number of variables, depending on $O$. More precisely, the
dividing line is given by the clones $\Ec$, $\Vc$ and $\Lc$.

\thmvc*

The proof uses the following basic fact about Post's lattice, which will
also be used in the next section.

\begin{fact}
\label{fact:minimal-hard}
For all clones  $\Cc$, the following are equivalent:
\begin{enumerate}
\item  $\Cc\not\subseteq \Ec$, $\Cc\not\subseteq \Vc$ and $\Cc\not\subseteq \Lc$;
\item $\Dc_2\subseteq \Cc$  or $\Sc_{00}\subseteq \Cc$ or $\Sc_{10}\subseteq \Cc$. Equivalently, $\Cc$ contains at least one of the three functions
\[
  \maj(x,y,z),
  \qquad
  f_{\vee\wedge}(x,y,z)=x\vee(y\wedge z),
  \qquad
  f_{\wedge\vee}(x,y,z)=x\wedge(y\vee z).
\]
\end{enumerate}
Moreover, if $\Cc$ contains $\maj$ but neither $f_{\vee\wedge}$ nor
$f_{\wedge\vee}$, then $\Cc\subseteq\Dc$.
\end{fact}

\begin{proof}[Proof of Proposition~\ref{prop:vc}]
For the upper bound, a class of VC dimension $d$ has at least $2^d$ members, so
it suffices to count.  The $n$-ary part of $\Ec=[\wedge,\top,\bot]$ consists of
the conjunctions $\bigwedge_{i\in S}x_i$ for $S\subseteq[n]$, with
$S=\varnothing$ giving $\top$, together with $\bot$, so it has $2^n+1$ members
and VC dimension at most $n$; dually for $\Vc$.  The $n$-ary part of
$\Lc=[\oplus,\top,\bot]$ consists of the functions
$c\oplus\bigoplus_{j\in J}x_j$ with $c\in\{0,1\}$ and $J\subseteq[n]$, so it
has $2^{n+1}$ members and VC dimension at most $n+1$.  Subclones only shrink
these classes.

For the lower bound, VC dimension is monotone in the class, so by
Fact~\ref{fact:minimal-hard} it suffices to exhibit, for $n\ge2$, a set of
size $2^{\Omega(n)}$ shattered by each of $\Dc_2$, $\Sc_{00}$ and $\Sc_{10}$.
We may assume that $n$ is even, since the $n$-ary part of a clone contains
every $(n-1)$-ary member with a dummy variable added, so that the VC
dimension in $n$ variables is at least that in $n-1$ variables.  Let
$m=(n-2)/2$ and let $A$ consist of the $2^m$ points
\[
  (1,\ \mathbf b,\ \bar{\mathbf b},\ 0)\in\{0,1\}^n,
  \qquad \mathbf b\in\{0,1\}^m .
\]
Every point of $A$ has first coordinate $1$, last coordinate $0$ and exactly
$m+1$ ones.  The argument has three steps.
First, $A$ is shattered by $\Mc$.  Its members have the same number of ones,
so none lies below another, and for $T\subseteq A$ the indicator function of
the upset generated by $T$ is monotone and is $1$ exactly on $T$ within $A$.
Second, on $A$ the coordinates $x_1$ and $x_n$ can stand in for the truth
constants, since $a_1=1$ and $a_n=0$ for every $\mathbf a\in A$.  Precisely,
let $\Cc$ be a clone and $f\in\Cc^{\pm}$.  As in the proof of
Lemma~\ref{lem:constant-adjunction}, $f(\mathbf x)=g(\mathbf x,0,1)$ for
some $g\in\Cc$, and $g(\mathbf x,x_n,x_1)$ is again a function of $\Cc$,
which agrees with $f$ on $A$.  Hence $A$ is shattered by $\Cc$ whenever it
is shattered by $\Cc^{\pm}$.
Third, $\Cc^{\pm}=\Mc$ for each of the three clones. Indeed,
$\maj(\bot,x,y)=x\wedge y$ and $\maj(\top,x,y)=x\vee y$ for $\Dc_2$,
$f_{\vee\wedge}(\bot,y,z)=y\wedge z$ and $f_{\vee\wedge}(x,y,y)=x\vee y$ for
$\Sc_{00}$, and dually for $\Sc_{10}$.  So $A$ is shattered by all three.
\end{proof}

\section{PAC learning}\label{sec:pac}

We make the learnability notions of the introduction precise.  For a
formula or circuit $\varphi$ over $n$ variables, let
$K(\varphi)=[\varphi]^{-1}(1)\subseteq\{0,1\}^n$ be the concept it defines,
and let $\size{\varphi}$ be the size of the representation.
\emph{Polynomially properly PAC learnable} is the first notion of the
introduction with these conventions.  In
\emph{polynomially PAC predictable with membership queries}, the weaker demand
of~\cite{AK95}, the predicted label must be wrong with probability at most
$1/2-1/p(s,n)$ for a polynomial $p$, and the learner must run in time polynomial
in $s$, $n$ and $1/\varepsilon$.

The VC bounds of Section~\ref{sec:vc}, with the fitting algorithm of
Section~\ref{sec:fitting}, already give the stronger,
representation-insensitive form of learnability: a single learner that
handles every target in $\PL_O$, with a number of examples polynomial in $n$,
$1/\varepsilon$ and $1/\delta$ but not permitted to grow with the size of the
target's representation.

\begin{proposition}[Learning $\PL_O$ uniformly]\label{prop:pac-uniform}
Let $O$ be a finite basis.  There is a polynomial-time learner for
$\PL_O$ using $\mathrm{poly}(n,1/\varepsilon,\log(1/\delta))$ examples
if and only if $O\preceq\{\wedge,\top,\bot\}$, $O\preceq\{\vee,\top,\bot\}$ or
$O\preceq\{\oplus,\top,\bot\}$.  The statement is unconditional, and in the
positive cases the hypothesis returned is a $\PL_O$-formula with at most $n+2$
leaves.
\end{proposition}

\begin{proof}
In the three positive cases, Proposition~\ref{prop:vc} bounds the VC dimension by
$n+1$, so a sample of size $O\bigl((n+\log(1/\delta))/\varepsilon\bigr)$
suffices for uniform convergence, and the fitting algorithm of
Theorem~\ref{thm:boolean-fitting-constructive} --- cases~(2) and~(4) of its
proof --- turns such a sample into a consistent $\PL_O$-formula in polynomial
time with at most $n+2$ leaves, namely a conjunction, a disjunction or a
parity of variables with at most two further leaves.  A consistent learner for a class of polynomial
VC dimension is a polynomial-time PAC learner.  The hypothesis is a formula, so
the positive half holds whether the learner must output a formula or may output
a circuit, and in particular for proper learning.

Conversely, outside the three cases Proposition~\ref{prop:vc} gives VC
dimension $2^{\Omega(n)}$, and any learner for a class of VC dimension $d$
requires $\Omega(d/\varepsilon)$ examples, hence exponential running time.  This holds however the learner represents its hypotheses.
\end{proof}

The positive half of Theorem~\ref{thm:pac} is contained in this proposition.
Its negative half is not: the obstruction above is information-theoretic and
evaporates once the running time and sample complexity  are allowed 
to depend on the size of the target concept. The remainder of this section
proves this direction. 

By the \emph{cryptographic assumptions} we mean throughout the assumption that
at least one of the following is intractable: testing quadratic residuosity
modulo a composite, inverting RSA encryption, and factoring Blum integers.
The starting point is the following theorem.

\begin{theorem}[{\cite{KV94,AK95}}]\label{thm:kv-ak}
Under the cryptographic assumptions, $\PL_{\{\wedge,\vee,\neg\}}$ is not
polynomially PAC predictable with membership queries.
\end{theorem}

Every conditional lower bound proved in this paper --- in this section and in
Sections~\ref{sec:occam} and~\ref{sec:noise} --- is obtained from
Theorem~\ref{thm:kv-ak} by reductions; the lower bounds of
Section~\ref{sec:erm-regions} rest instead on $\Ptime\ne\NP$ or on the Unique
Games Conjecture.  The reductions are of a type due to Angluin and
Kharitonov~\cite{AK95}, a membership-query variant of the
prediction-preserving reductions of Pitt and Warmuth~\cite{PW90}.  Below
$\mathcal F,\mathcal F'$ are fragments such as $\PL_O$ or $\CIR_O$, and
$\varphi$ ranges over the members of $\mathcal F$.

\begin{definition}[{\cite[Definition~1]{Dalmau}, after~\cite{AK95,PW90}}]
\label{def:pwm}
$\mathcal F\pwm\mathcal F'$ if there are mappings $g$ (formulas), $\iota$
(instances) and $h,j$ (queries), with $g(s,n,\varphi)\in\mathcal F'$, such
that
\begin{enumerate}[label=(\arabic*)]
\item there is a nondecreasing polynomial $q$ with
      $\size{g(s,n,\varphi)}\le q(s,n,\size{\varphi})$
      for all $s,n$ and all $\varphi$ with $\size{\varphi}\le s$;
\item $\iota$ is computable in time polynomial in $s,n,|\mathbf w|$, and
      $\mathbf w\in K(\varphi)$ iff
      $\iota(s,n,\mathbf w)\in K(g(s,n,\varphi))$ whenever
      $\size{\varphi}\le s$ and $\mathbf w\in\{0,1\}^n$;
\item $h$ and $j$ are computable in time polynomial in $s,n,|\mathbf w'|$, and
      for \emph{every} tuple $\mathbf w'$: $\mathbf w'\in K(g(s,n,\varphi))$
      iff $j(s,n,\mathbf w',b)=\top$, where $b:=\top$ if
      $h(s,n,\mathbf w')\in K(\varphi)$ and $b:=\bot$ otherwise.
\end{enumerate}
\end{definition}

The relevant property of these reductions is that if
$\mathcal F\pwm\mathcal F'$ and $\mathcal F'$ is polynomially PAC predictable
with membership queries, then so is $\mathcal F$ (cf.~\cite[Lemmas~1 and~2]{Dalmau}).

We  now prove Theorem~\ref{thm:pac}, restated at the end of this section. 
Consider $\PL_O$ where $O$ falls under none of the three positive cases.  
In that case
$[O]$ contains one of
\[
  f_{\vee\wedge}(x,y,z)=x\vee(y\wedge z),
  \qquad
  f_{\wedge\vee}(x,y,z)=x\wedge(y\vee z),
  \qquad
  \maj(x,y,z),
\]
by Fact~\ref{fact:minimal-hard} and the reductions
to be constructed are those of Theorems~2--4 of~\cite{Dalmau}.
Those reductions substitute a fixed gadget for each gate of the source formula.
In a circuit this costs a constant factor, the gadget sharing its argument
wires.  In a formula it need not, since the substitution patterns duplicate
arguments and the gadget may read each placeholder more than once, so size can
grow by a constant factor per level.  The published argument therefore yields
$\PL_{\{\wedge,\vee,\neg\}}\pwm\CIR_O$, which is weaker than
Theorem~\ref{thm:pac}: for a fixed size bound the formula-represented concepts
are a subclass of the circuit-represented ones, and hardness does not pass to
subclasses.  The missing step is to rebalance the source formula before
substituting, since on a formula of logarithmic depth the substitution costs a
polynomial factor overall.  Balancing introduces the truth constants, which
$\PL_O$ need not have.  They are carried instead by the auxiliary variables that
the instance mapping already appends, at the cost of one further such variable
in two of the three cases.  We record Spira's theorem in the form used.

\begin{fact}[Balancing, \cite{Spira}]\label{fact:spira}
Let $\varphi$ be a formula over a finite set of Boolean connectives.  There is
an equivalent $\PL_{\{\wedge,\vee,\neg\}}$-formula of depth
$O(\log\size{\varphi})$ and of size polynomial in $\size{\varphi}$, computable
from $\varphi$ in polynomial time.
\end{fact}

The next lemma collects everything that happens while the truth constants are
still available, and does so once and for all: its hypothesis is only that $[O]$
contains the monotone clone $\Mc=[\wedge,\vee,\top,\bot]$, which holds for every
basis of $\Mc$ and, more to the point below, for $O\cup\{\top,\bot\}$ whenever
$[O]$ contains one of the three functions above.  Since the basis is arbitrary,
all the construction needs of $O$ are fixed formulas for $\wedge$, $\vee$ and
the two truth constants.

\begin{lemma}\label{lem:monotone-hard}
Let $O$ be a finite set of Boolean functions such that $\Mc\subseteq [O]$.  Then $\PL_{\{\wedge,\vee,\neg\}}\pwm\PL_O$.
\end{lemma}

\begin{proof}
Let $\varphi$ be a $\PL_{\{\wedge,\vee,\neg\}}$-formula over the variables
$x_1,\ldots,x_n$.  By Fact~\ref{fact:spira} we may assume that
$\dep(\varphi)$ is bounded logarithmically in $\size{\varphi}$, at the cost of
replacing $\varphi$ by an equivalent formula of size polynomial in $\size{\varphi}$,
computable in polynomial time.  The balanced formula may contain the truth
constants, which costs nothing below, both lying in $\Mc$.  The remainder follows Lemma~3
of~\cite{Dalmau}.  For the sake of completeness, we spell out the details.

Pushing negations to the leaves by de Morgan's laws turns $\varphi$ into a
monotone formula $\varphi^{+}$ over $\{\wedge,\vee,\top,\bot\}$ and the $2n$
variables $x_1,\ldots,x_n,y_1,\ldots,y_n$, the variable $y_i$ taking the place
of $\neg x_i$, so that
$\varphi^{+}(\mathbf a,\bar{\mathbf a})=\varphi(\mathbf a)$ for every
$\mathbf a\in\{0,1\}^n$.  The tree is unchanged, so $\varphi^{+}$ has the same depth as
$\varphi$.  On its own $\varphi^{+}$ is of no use, because condition~(3) of
Definition~\ref{def:pwm} ranges over all query tuples, including those in
which the $y_i$ are not set to the complements of the $x_i$, where the value of $\varphi^{+}$ need not be
determined by $\varphi$ at all.  Two guards repair this.  Let
\[
  A:=\bigvee_{i\le n}(x_i\wedge y_i),
  \qquad
  B:=\bigwedge_{i\le n}(x_i\vee y_i),
  \qquad
  \psi:=A\vee(\varphi^{+}\wedge B),
\]
with $A$ and $B$ written as balanced trees, so that $\dep(A)$ and $\dep(B)$ are
at most $\lceil\log n\rceil+1$.  Then $\psi$ is again a monotone formula over
$\{\wedge,\vee,\top,\bot\}$.  Its depth is
$\max\{\dep(\varphi^{+}),\lceil\log n\rceil+1\}+2$, hence still logarithmic in
$\size{\varphi}+n$, and its value is determined by $\varphi$ at every point.
Indeed, write $\langle\mathbf a,\mathbf a'\rangle$ for the setting that gives
$x_1,\ldots,x_n$ the values $\mathbf a$ and $y_1,\ldots,y_n$ the values
$\mathbf a'$.  If $a_i=a'_i=1$ for some $i$ then $A$ holds and $\psi$ evaluates
to $1$; if $a_i=a'_i=0$ for some $i$ and the previous case does not apply then
$B$ fails and $\psi$ evaluates to $0$; and otherwise
$\mathbf a'=\bar{\mathbf a}$, where $A$ fails, $B$ holds, and $\psi$ evaluates
to $\varphi(\mathbf a)$.

The four connectives of $\psi$ lie in $\Mc\subseteq[O]$.  Fix a $\PL_O$-formula
computing each of them, and let $\chi$ be the result of substituting these into
$\psi$, so that $\chi\in\PL_O$ is equivalent to $\psi$.  Substituting a fixed
formula for each node multiplies depth by a constant, so $\dep(\chi)$ too is
logarithmic in $\size{\varphi}+n$, and a formula of depth $d$ whose connectives
all have arity at most $r$ has at most $r^{d}$ leaves, so
$\size{\chi}=(\size{\varphi}+n)^{O(1)}$, with an exponent depending only on the
four formulas just chosen.

It remains to collect the mappings.  They are:
\begin{itemize}
\item $g(s,n,\varphi):=\chi$;
\item $\iota(s,n,\mathbf a):=\langle \mathbf a,\bar{\mathbf a}\rangle$;
\item $h(s,n,\langle \mathbf a,\mathbf a'\rangle):=\mathbf a$, and
\item $j(s,n,\langle \mathbf a,\mathbf a'\rangle,b):=\top$ if $a_i=a'_i=1$ for
some $i$, $:=\bot$ if not and $a_i=a'_i=0$ for some $i$, and $:=b$ otherwise.
\end{itemize}
Condition~(1) of Definition~\ref{def:pwm} is the size bound just proved,
condition~(2) holds because
$\chi(\mathbf a,\bar{\mathbf a})=\varphi(\mathbf a)$, and condition~(3) is the
case distinction of the previous paragraph.
\end{proof}

The remainder of the proof is exactly as in \cite{Dalmau}. For completeness,
we spell it out below.

\begin{lemma}
\label{lem:join-case}
Let $O$ be a finite set of Boolean functions with $f_{\vee\wedge}\in[O]$ or
$f_{\wedge\vee}\in[O]$.  Then $\PL_{\{\wedge,\vee,\neg\}}\pwm\PL_O$.
\end{lemma}

\begin{proof}
We give the proof for the case $f_{\vee\wedge}\in[O]$. The argument for the
other case is dual.

Since conjunction and disjunction are definable as $f_{\vee\wedge}(\bot,x,y)$
and $f_{\vee\wedge}(x,y,y)$, respectively, we have that 
 $\Mc\subseteq[O\cup\{\top,\bot\}]$ and Lemma~\ref{lem:monotone-hard} gives
$\PL_{\{\wedge,\vee,\neg\}}\pwm\PL_{O\cup\{\top,\bot\}}$.  By transitivity it remains
to remove the truth constants, that is, to show
$\PL_{O\cup\{\top,\bot\}}\pwm\PL_O$.
Recall that a pwm-reduction consists of mappings $g$ (for concepts), $\iota$
(for instances) and $h,j$ (for queries). The mappings in question are:
\begin{itemize}
\item  $g(s,n,\chi) := \varphi_{\vee\wedge}(z_0,z_1,\chi'(\mathbf x,z_0,z_1))$ 

where 
 $\varphi_{\vee\wedge}$ is a $\PL_O$-formula defining $f_{\vee\wedge}$,
 and where $\chi'\in\PL_O$ is obtained from $\chi$ by replacing $\bot$  by a fresh variable $z_0$ and  $\top$ by a fresh variable $z_1$.
 \item $\iota(s,n,\mathbf a):=\langle \mathbf a,0,1\rangle$.
\item  $h(s,n,\langle \mathbf a,c_0,c_1\rangle):=\mathbf a$, and
\item
$j(s,n,\langle \mathbf a,0,1\rangle,b):=b$

$j(s,n,\langle \mathbf a,0,0\rangle,b):=\bot$

$j(s,n,\langle \mathbf a,1,0\rangle,b):=\top$

$j(s,n,\langle \mathbf a,1,1\rangle,b):=\top$
\end{itemize}
Here $\mathbf a$ ranges over $\{0,1\}^n$ and $c_0,c_1$ over $\{0,1\}$, the
query string $\langle \mathbf a,c_0,c_1\rangle$ giving $x_1,\ldots,x_n$ the
values $\mathbf a$ and $z_0,z_1$ the values $c_0,c_1$.
The conditions of Definition~\ref{def:pwm} are met: $g(\chi)$ computes
$z_0\vee(z_1\wedge\chi'(\mathbf x,z_0,z_1))$, which takes the value
$\chi(\mathbf a)$ at $\langle \mathbf a,0,1\rangle$, the value $0$ at
$\langle \mathbf a,0,0\rangle$ and the value $1$ whenever $c_0=1$ --- in every
case but the first, independently of what $\chi'$ computes there.  The two
formulas differ in size by the one copy of $\varphi_{\vee\wedge}$.
\end{proof}

\begin{lemma}
\label{lem:selfdual-case}
Let $O$ be a finite set of Boolean functions with $\maj\in[O]$, all of whose
members are self-dual; that is, let $\Dc_2\subseteq[O]\subseteq \Dc$.  Then
$\PL_{\{\wedge,\vee,\neg\}}\pwm\PL_O$.
\end{lemma}

\begin{proof}
Since conjunction and disjunction are definable as $\maj(\bot,x,y)$ and
$\maj(\top,x,y)$, respectively, we have that $\Mc\subseteq[O\cup\{\top,\bot\}]$
and Lemma~\ref{lem:monotone-hard} gives
$\PL_{\{\wedge,\vee,\neg\}}\pwm\PL_{O\cup\{\top,\bot\}}$.  By transitivity it remains
to remove the truth constants, that is, to show
$\PL_{O\cup\{\top,\bot\}}\pwm\PL_O$.  The mappings are:
\begin{itemize}
\item  $g(s,n,\chi) := \varphi_{\maj}(\chi'(\mathbf x,z_0,z_1),z_0,z_1)$

where
 $\varphi_{\maj}$ is a $\PL_O$-formula defining $\maj$,
 and where $\chi'\in\PL_O$ is obtained from $\chi$ by replacing $\bot$  by a fresh variable $z_0$ and  $\top$ by a fresh variable $z_1$.
 \item $\iota(s,n,\mathbf a):=\langle \mathbf a,0,1\rangle$.
\item  $h(s,n,\langle \mathbf a,c_0,c_1\rangle):=\mathbf a$ if
$(c_0,c_1)\neq(1,0)$ and $:=\bar{\mathbf a}$ otherwise, where
$\bar{\mathbf a}$ is the bitwise complement of $\mathbf a$, and
\item
$j(s,n,\langle \mathbf a,0,1\rangle,b):=b$

$j(s,n,\langle \mathbf a,0,0\rangle,b):=\bot$

$j(s,n,\langle \mathbf a,1,0\rangle,b):=\neg b$

$j(s,n,\langle \mathbf a,1,1\rangle,b):=\top$
\end{itemize}
Again the conditions of Definition~\ref{def:pwm} are met:
$g(\chi)$ computes $\maj(\chi'(\mathbf x,z_0,z_1),z_0,z_1)$, which is $0$ at
$(c_0,c_1)=(0,0)$ and $1$ at $(1,1)$ independently of what $\chi'$ computes
there, and which is the value of $\chi'$ itself at the two remaining settings
--- at $(0,1)$ that value is $\chi(\mathbf a)$, and at
$(1,0)$ it is $\neg\chi(\bar{\mathbf a})$, since evaluating $\chi'$ at
$(1,0)$ evaluates $\chi$ with $\top$ and $\bot$ interchanged, and a formula all
of whose connectives are self-dual then computes the dual function.
\end{proof}

\thmpac*

\begin{proof}
That~(a) implies~(b) and~(c) is
Proposition~\ref{prop:pac-uniform} together with the remark after it: the
learner there uses no promise on the target's size, and it certainly yields both
a PAC learner and a PAC predictor.  That~(b) implies~(c) is immediate, PAC
prediction being easier than producing a hypothesis and membership queries
only adding power.  It remains to show that~(c) implies~(a), which we do by contraposition.
Suppose $O$ satisfies none of the three conditions.  By
Fact~\ref{fact:minimal-hard} either $[O]$ contains $f_{\vee\wedge}$ or
$f_{\wedge\vee}$, and then $\PL_{\{\wedge,\vee,\neg\}}\pwm\PL_O$ by
Lemma~\ref{lem:join-case}; or it contains neither, and then $[O]\subseteq \Dc$
while still containing one of the three generators, necessarily $\maj$, so that
$\PL_{\{\wedge,\vee,\neg\}}\pwm\PL_O$ by Lemma~\ref{lem:selfdual-case}.  Either way,
Theorem~\ref{thm:kv-ak} shows that $\PL_O$ is not polynomially PAC predictable with
membership queries.  Finally $\PL_O\pwm\CIR_O$ by the identity
reduction --- a formula is a tree-shaped circuit of the same size --- so the
circuit class is not polynomially PAC predictable either, and the equivalence holds
for it as well.
\end{proof}

\section{Occam algorithms}\label{sec:occam}

Recall from the introduction that an Occam algorithm for $\PL_O$ is a
polynomial-time fitting algorithm returning, on every realizable sample, a
fitting formula of size at most
$p\bigl(s_{\mathrm{opt}}(E),n\bigr)\cdot m^{\beta}$ for a fixed polynomial $p$
and a fixed $\beta<1$, attribute-efficient if the bound does not depend on $n$.
Size may be read as formula or as circuit size, the positive results producing
formulas and the negative one being proved for circuits.  The negative
direction is a corollary of Theorem~\ref{thm:pac}: a sufficiently compressing
fitting algorithm is a proper PAC learner, so a class that is not even
polynomially PAC predictable cannot have one.

\thmoccam*

\begin{proof}
For~(a), all three cases are handled by the standard greedy, set-cover-inspired
fitting algorithm; see for instance~\cite{KearnsVazirani}.  If
$O\preceq\{\wedge,\top,\bot\}$, the fitting hypotheses are the conjunctions
$\bigwedge_{i\in S}x_i$ with $S$ contained in the set $T$ of coordinates that
are $1$ on every positive example and meeting, for every negative example, the
set of coordinates of $T$ on which it is $0$.  Minimizing $|S|$ is a hitting set
problem with at most $m$ sets to hit, and greedy returns a set of size at most
$(1+\ln m)$ times the minimum, so the returned formula has size
$O(s_{\mathrm{opt}}(E)\log m)$.  When the empty conjunction $\top$ is
unavailable, as for $\Ec_2=[\wedge]$, replace it by a single element of $T$, which
is nonempty on a realizable sample.  The case $O\preceq\{\vee,\top,\bot\}$ is
dual, and if $O\preceq\{\neg,\top,\bot\}$ then every hypothesis is a constant, a
variable or a negated variable, so any fitting hypothesis has size $O(1)$ and
the fitting algorithm of Section~\ref{sec:fitting} returns one.  In all three
cases the bound is independent of $n$.

For~(b), case~(2) of the proof of
Theorem~\ref{thm:boolean-fitting-constructive} returns, by Gaussian
elimination, a fitting affine function $c\oplus\bigoplus_{j\in J}x_j$ written
as a $\PL_O$-formula with at most $n+2$ leaves.  That is a bound of the
required form, with $\beta=0$, but it depends on $n$: what is not known is how
to control $|J|$, and with it the size of the formula, in terms of
$s_{\mathrm{opt}}(E)$, and we leave that open.

For~(c), an Occam algorithm for $\PL_O$ is a proper PAC
learner~\cite{Blumer}, so $\PL_O$ would be polynomially properly PAC
learnable, which Theorem~\ref{thm:pac} excludes outside the regions of~(a)
and~(b) under the cryptographic assumptions.
\end{proof}

\section{Random label noise and statistical queries}\label{sec:noise}

Fix a target $f$ and a distribution $D$ on $\{0,1\}^n$.  Under \emph{random
classification noise} of rate $\eta<1/2$~\cite{AngluinLaird} the learner
receives examples $(\mathbf a,b)$ with $\mathbf a\sim D$ and $b=f(\mathbf a)$
flipped independently with probability $\eta$, and must reach error
$\varepsilon$ with probability $1-\delta$ in time polynomial in $n$, $s$,
$1/\varepsilon$, $1/\delta$ and $1/(1-2\eta_b)$, where $\eta_b<1/2$ is a given
upper bound on $\eta$.  In the \emph{statistical query} model~\cite{Kearns98}
the learner sees no examples; it may ask, for any polynomial-time predicate
$\chi(\mathbf a,b)$ and any tolerance $\tau$, for a number within $\tau$ of
$\Pr_{\mathbf a\sim D}[\chi(\mathbf a,f(\mathbf a))=1]$, and it is
\emph{efficient} if it uses polynomially many queries, each of tolerance at
least $1/p(n,s,1/\varepsilon)$ for a fixed polynomial $p$, and polynomial
time.  Kearns proved that such a learner can be simulated from examples,
noisy or not: every class efficiently learnable from statistical queries is
PAC learnable under random classification noise of any rate $\eta<1/2$, in
polynomial time~\cite{Kearns98}.

\thmsq*

The positive direction is Kearns's statistical-query algorithm for
conjunctions~\cite{Kearns98}; disjunctions are dual, and for
$O\preceq\{\neg,\top,\bot\}$ there are at most $2n+2$ hypotheses, whose
errors can be estimated one by one.
For the affine bases the negative direction is unconditional and holds
already under the uniform distribution: $[O]\supseteq\Lc_2$ contains the
$2^{n-1}$ parities of an odd number of variables, which are pairwise
uncorrelated under that distribution, and a class with exponentially many
pairwise uncorrelated members needs exponentially many statistical queries or
exponentially small tolerance~\cite{Kearns98,BFJKMR}.  For every other basis
outside the three cases, a statistical-query learner would yield a PAC
learner and hence a PAC predictor, so $\PL_O$ would be polynomially PAC predictable
with membership queries, which Theorem~\ref{thm:pac} excludes under the
cryptographic assumptions.  

By Kearns's simulation, the positive direction gives polynomial-time PAC
learning under random classification noise of any rate $\eta<1/2$ for the
conjunctive, the disjunctive and the essentially unary bases.  The negative
direction transfers only in part.  Outside the PAC-learnable region noise is
beside the point: a learner that tolerates noise of rate up to $\eta_b$ can be
run on noise-free examples, so Theorem~\ref{thm:pac} excludes it under the
cryptographic assumptions.  For the affine clones, however, the
statistical-query lower bound says nothing about noisy examples, and it is
exactly here that the two models are known to part company: Blum, Kalai and
Wasserman~\cite{BKW} learn parities under random classification noise in time
$2^{O(n/\log n)}$, fewer than the $2^{\Omega(n)}$ statistical queries needed
above, and whether polynomial time is possible --- the learning parity with
noise problem --- is open.  

\section{Other kinds of fragments}\label{sec:constraints}

Every fragment considered above is generated by a set of connectives, and is
therefore closed under substitution.
Several fragments of independent interest are not of this shape.  These include
monotone CNF, Horn CNF, and systems of linear equations over $\Ftwo$.  
None of the classifications of the introduction applies to such fragments.
In this section, we consider a few other types of fragments.

\subsection{Fragments given by constraint languages}

A \emph{constraint language} (also known as a \emph{template})
is a set $\Gamma$ of Boolean relations.  A
$\CNF_\Gamma$-formula over $x_1,\ldots,x_n$ is a conjunction of atoms
$R(x_{i_1},\ldots,x_{i_k})$ with $R\in\Gamma$ of arity $k$; the concept it
defines is its set of satisfying assignments in $\{0,1\}^n$, and its size is its
number of atoms.  Deciding whether a $\CNF_\Gamma$-formula has a model is the
well-studied
\emph{constraint satisfaction problem} $\CSP(\Gamma)$, which was classified by Schaefer's
theorem~\cite{Schaefer}. 

Monotone CNF, Horn CNF and systems of linear equations can all be viewed
as instances of $\CNF_\Gamma$ for different choices of the constraint
language $\Gamma$. For instance, Horn CNF is $\CNF_\Gamma$ for 
$\Gamma$ the set of all Horn clauses. In each of these cases, $\Gamma$
is infinite, but there are also natural fragments $\CNF_\Gamma$ where
$\Gamma$ is finite, such as $k$-CNF, for a fixed value of $k$, where $\Gamma$ is 
the finite set of all $k$-ary relations over the Booleans.  $\PL_{\{\wedge,\top\}}$ 
can also be cast as $\CNF_{\{\mathsf{True}\}}$,
where $\mathsf{True}$ is the unary Boolean relation containing only the tuple $(1)$.

For finite constraint languages $\Gamma$, results analogous to the
classifications of the introduction hold, empirical risk minimization aside.
In fact the situation here is simpler: for every finite $\Gamma$ there is a
fitting formula within a logarithmic factor of the smallest, and the VC
dimension of $\CNF_\Gamma$ is polynomial in the number of variables.
Note that, since there are only finitely many Boolean relations
of any given arity, every constraint language of bounded arity is finite, 
up to logical equivalence, and conversely a finite constraint language 
has bounded arity.

\begin{proposition}\label{prop:cnf-tractable}
Fix a finite constraint language  $\Gamma$ and let $r$ be the largest arity of a
relation in $\Gamma$.  Then
\begin{enumerate}[label=(\arabic*)]
\item the existence of a fitting $\CNF_\Gamma$-formula for a given labeled sample is decidable in polynomial time,
\item given a
      realizable labeled sample $E$ of size $m$,
      one may in addition compute in polynomial time a fitting
      $\CNF_\Gamma$-formula of size at most $(1+\ln m)\cdot s_{\mathrm{opt}}(E)$,
      where $s_{\mathrm{opt}}(E)$ is the number of atoms of a smallest fitting
      formula,
\item the VC dimension of $\CNF_\Gamma$ in $n$ variables is at most $|\Gamma|\,n^r$, and hence,
      together with~(2), $\CNF_\Gamma$-formulas in $n$ variables are
      polynomially properly PAC learnable.
\end{enumerate}
\end{proposition}

\begin{proof}
First, observe that there are at most $|\Gamma|\,n^r$ atoms over $x_1,\ldots,x_n$.
For~(1), let $A$ be the set of atoms satisfied by every
positive example of $E$.  Any $\CNF_\Gamma$-formula whose concept contains all
positive examples uses only atoms from $A$, so $\bigwedge A$ defines the
$\subseteq$-least concept of $\CNF_\Gamma$ containing them.  The sample is
realizable if and only if no negative example satisfies $\bigwedge A$, and
$\bigwedge A$ is then a fitting formula.  For (2), observe that a
$\CNF_\Gamma$-formula fits $E$ if and only if its atoms lie in $A$ and every
negative example violates at least one of them: this is a set cover instance
with the negative examples as ground set and $A$ as the family of sets, so the
greedy set-cover algorithm~\cite{Hochbaum} returns a fitting formula within a factor
$1+\ln m$ of the smallest.  For~(3), a concept of $\CNF_\Gamma$ is determined by
a set of atoms, so the class has at most $2^{|\Gamma|n^r}$ members and its VC
dimension is at most $\log_2$ of that.  Finally, the polynomial VC dimension together
with~(2) gives proper PAC learning.
\end{proof}

Note also that the Occam bound in~(2), measured
in number of atoms, does not depend on $n$, so the algorithm can be said to be
attribute-efficient in the sense of Section~\ref{sec:occam}.
The fitting problem for $\CNF_\Gamma$ is the \emph{structure identification}
problem of Dechter and Pearl~\cite{DechterPearl}, studied for general $\Gamma$
by Creignou, Kolaitis and Zanuttini~\cite{CKZ}.  Its variant in which only the
positive examples are given, and the formula must define exactly their set, is
the inverse satisfiability problem of Kavvadias and
Sideri~\cite{KavvadiasSideri,LagerkvistWahlstrom}.

When it comes to \emph{empirical risk minimization}, hardness cases arise.
Recall that  $\PL_{\{\wedge,\top\}}$ coincides with $\CNF_\Gamma$ for 
$\Gamma=\{\mathsf{True}\}$ as described above. Therefore,
Theorem~\ref{thm:erm}\,(i) applies, and there is no polynomial-time weak
approximator for $\CNF_{\{\mathsf{True}\}}$ unless $\Ptime=\NP$.  The lower bound there is
representation-independent, so it does not matter that the hypothesis is now
written as a conjunction of atoms.  
However, we do not know where the dividing
line runs, either for exact solvability or for approximability of ERM.

For \emph{infinite} constraint languages $\Gamma$,
Proposition~\ref{prop:cnf-tractable} fails as stated. Note that there are then
relations of unbounded arity, and the atom count over $n$ variables is no longer
polynomial.  Four examples show what can happen.  

\begin{example}[Monotone CNF]
The concepts of monotone CNF over $n$ variables are exactly the upsets of
$\{0,1\}^n$, since a monotone function is the conjunction of its prime
implicates and these are positive clauses.  Fitting is polynomial: the least
concept containing the positive examples is the upset they generate, so the
sample is realizable if and only if no negative example lies above a positive
one.  The VC dimension, on the other hand, is $\binom{n}{\lfloor n/2\rfloor}$, hence
exponential, as follows from Sperner's theorem~\cite{Sperner}.  Finally, without membership
queries, 
PAC-learning is no easier for monotone CNF than for arbitrary
CNF~\cite{KLPV}.  In particular, monotone CNF is not polynomially properly PAC learnable
from random examples alone unless $\NP=\RP$~\cite{ABFKP}.  
Monotone CNF \emph{is} polynomially properly PAC learnable with membership
queries~\cite{Angluin88}.
\end{example}

\begin{example}[Horn CNF] Fitting is again polynomial, in both its decision and its
construction form~\cite{DechterPearl}. Since every monotone clause is a dual Horn
clause, which can be turned into a Horn clause by complementing all coordinates,
the VC dimension $\binom{n}{\lfloor n/2\rfloor}$ for monotone CNF is also a lower
bound on the VC dimension of Horn CNF, and the hardness of \emph{proper} PAC
learning from random examples transfers from monotone CNF to Horn CNF as
well. Horn CNF \emph{is} polynomially properly PAC learnable with
membership queries~\cite{AFP}.
\end{example}

\begin{example}[Systems of linear equations]
Here $\Gamma_{\mathrm{lin}}$ consists of the relations
$x_{i_1}\oplus\cdots\oplus x_{i_k}=c$ for $k\ge0$ and $c\in\{0,1\}$, again an
infinite template, and the concepts of $\CNF_{\Gamma_{\mathrm{lin}}}$ over $n$
variables are the affine subspaces of $\Ftwo^n$ together with $\varnothing$.
Unlike the previous two examples, this one is well behaved throughout.  Fitting
is polynomial in both its decision and its construction form: the least concept
containing the positive examples is their affine hull, computed by Gaussian
elimination, the sample is realizable exactly when no negative example lies in
it, and at most $n$ equations define it, so the fitting formula has at most $n$
atoms.  The VC dimension is $n+1$, a set being shattered exactly when it is
affinely independent.  As a result, $\CNF_{\Gamma_{\mathrm{lin}}}$ is
polynomially properly PAC learnable~\cite{HSW}.  
\end{example}

\begin{example}[A simple infinite template that is hard for fitting]
\label{example:infinite-fitting-hard}
Let $\Gamma_{\mathrm{one}}=\{\mathsf{Uniq}_k : k\ge1\}$, where
$\mathsf{Uniq}_k(x_1,\ldots,x_k)$ holds when exactly one of its arguments is
$1$.  Then fitting for $\CNF_{\Gamma_{\mathrm{one}}}$ is $\NP$-hard, 
as can be shown by a reduction from exact cover by $3$-sets~\cite{GJ}: given a set
$U$ and a family $\mathcal{S}$ of $3$-element subsets of $U$, is there a subfamily
covering each element of $U$ exactly once?  Take one variable $x_S$ for each
$S\in \mathcal{S}$, one positive example $\mathbf p_u$ for each $u\in U$, namely the
characteristic vector of $\{S\in \mathcal{S} : u\in S\}$, and the single negative example
$\zero$.
If $\mathcal{S}'=\{S_1, \ldots, S_k\}$ is an exact cover of size $k$, the single-atom formula 
$\mathsf{Uniq}_k(x_{S_1}, \ldots, x_{S_k})$
fits: exactly one argument is $1$ under $\mathbf p_u$, since exactly one member
of $\mathcal{S}'$ contains $u$, while none is $1$ under $\zero$.  Conversely,
suppose some formula fits.  It has at least one atom, since the empty
conjunction accepts $\zero$.  Let $\mathsf{Uniq}_k(x_{S_1}, \ldots, x_{S_k})$ be an arbitrary
atom in the conjunction.  Its arguments are pairwise distinct: each $S\in
\mathcal{S}$ is nonempty, so a repeated variable $x_S$ would give two true
arguments under $\mathbf p_u$ for any $u\in S$.  As every $\mathbf p_u$
satisfies the atom, each $u\in U$ lies in exactly one of $S_1,\ldots,S_k$, that
is, $\{S_1, \ldots, S_k\}$ is an exact cover.  
\end{example}

\subsection{Adding existential quantification}

A classification of $\CNF_\Gamma$  indexed by Post's lattice is not
immediately available.  The expressive power of $\CNF_\Gamma$ depends on $\Gamma$  up to
definability by quantifier-free conjunctive formulas, which is strictly finer
than primitive positive definability --- definability by conjunctions of atoms
and equalities with existential quantification, the closure operation of the
next paragraph --- so that the usual Galois connection with the polymorphisms
$\Pol(\Gamma)$ defined there does not apply.%
\footnote{The fact that the expressive power of $\CNF_\Gamma$ does not solely
depend on the polymorphisms of $\Gamma$, can be seen by taking
$\Gamma=\{\oplus^3(x,y,z)=0\}$ and $\Gamma'=\Gamma\cup\{{=}\}$.
The two templates have the
same polymorphisms, since every operation preserves equality.
But the $\CNF_{\Gamma'}$-formula $x=y$ is not expressible in
$\CNF_\Gamma$.}
There is, therefore, no reason to believe that, for example, the complexity of ERM
for finite-template $\CNF_\Gamma$ would be determined by the
polymorphisms of $\Gamma$.  What does govern quantifier-free conjunctive
definability is the \emph{partial} polymorphisms of
$\Gamma$~\cite{SchnoorSchnoor}, and the lattice of strong partial clones they
give rise to is far more complicated than Post's~\cite{AlekseevVoronenko}.
Allowing existential
quantification fixes this and restores the connection to Post's lattice.

An $\exCNF_\Gamma$-formula over
$x_1,\ldots,x_n$ is a formula
\[
  \exists z_1\cdots\exists z_\ell\ \varphi(x_1,\ldots,x_n,z_1,\ldots,z_\ell),
\]
where $\varphi$ is a conjunction of atoms $R(\cdot)$ with $R\in\Gamma$ and of
equalities between variables; the concept it defines is the set of
$\mathbf a\in\{0,1\}^n$ that extend to a satisfying assignment of $\varphi$, and
its size is the number of atoms of $\varphi$.  The concepts so definable are
exactly the $n$-ary relations of the co-clone
$\langle\Gamma\rangle=\Inv(\Pol(\Gamma))$ generated by $\Gamma$.  Here an
operation $f$ of arity $k$ \emph{preserves} a relation $R$ if applying $f$
coordinatewise to $k$ tuples of $R$ again gives a tuple of $R$.  Then $\Pol(\Gamma)$
is the clone of operations preserving every relation of $\Gamma$, and $\Inv(\Cc)$
is the set of relations preserved by every operation of $\Cc$.  A set of relations
of the form $\Inv(\Cc)$ is a \emph{co-clone}, and $\langle\Gamma\rangle$ is the
least co-clone containing $\Gamma$, namely the closure of $\Gamma$ under exactly
the constructs used above: conjunction, existential quantification, equality and
identification of variables.  Co-clones are to conjunctive existential
definability what clones are to substitution, and $\Pol$ and $\Inv$ match the
two lattices antitonically. We briefly describe the situation for two of our main questions
in this setting.

\emph{Fitting.}  The existence of a fitting $\exCNF_\Gamma$-formula
can be tested in polynomial time, regardless of the choice of $\Gamma$.  The
concepts of $\exCNF_\Gamma$ over $n$ variables are the subsets of $\{0,1\}^n$
preserved by $\Pol(\Gamma)$, so the closure of the set $P$ of positive
examples under the operations of $\Pol(\Gamma)$, applied coordinatewise, is
the least concept containing $P$, and the sample is realizable exactly when no
negative example lies in that closure.  That is again a two-element subpower
membership problem in the algebraic form of Section~\ref{sec:fitting}, with
the positive examples in place of the columns of the variables and a negative
example in place of the column of labels, and it is therefore polynomial by
the same case analysis over Post's lattice as in the proof of
Theorem~\ref{thm:boolean-fitting-constructive}, applied to $\Pol(\Gamma)$.
Unlike there, this route does not provide a fitting $\exCNF_\Gamma$-formula of
polynomial size: the witness it produces is a term of $\Pol(\Gamma)$, not an
$\exCNF_\Gamma$-formula.

\emph{PAC learning.}  For finite $\Gamma$ the learnability of $\exCNF_\Gamma$
was classified by Dalmau~\cite{Dalmau99} and rederived from the
polymorphisms of $\Gamma$ by Dalmau and Jeavons~\cite[Theorem~15]{DalmauJeavons}.
In the terminology of Section~\ref{sec:pac}, exactly one of the following
holds.
\begin{enumerate}[label=(\alph*)]
\item $\Pol(\Gamma)$ contains a near-unanimity operation or the affine
      operation $x\oplus y\oplus z$.  Then $\exCNF_\Gamma$ is polynomially
      PAC predictable, even without membership queries, and in the
      near-unanimity case it is polynomially properly PAC
      learnable~\cite[Corollary~1]{DalmauJeavons}.
\item Otherwise $\exCNF_\Gamma$ is not polynomially PAC predictable with
      membership queries, under the cryptographic assumption
      of~\cite{DalmauJeavons}, the existence of public-key cryptosystems
      secure against chosen-ciphertext attack.
\end{enumerate}
Dalmau and Jeavons state the positive half as polynomial learnability from
equivalence queries, with $\exCNF_\Gamma$-formulas as hypotheses in the
near-unanimity case.  A polynomial equivalence-query learner yields a
polynomial PAC learner with the same hypotheses~\cite{Angluin88}, which is
proper when these are $\exCNF_\Gamma$-formulas and in any case a PAC predictor,
the hypotheses being evaluable in polynomial time.  In the affine case their
learner is not proper, and proper PAC learnability is not addressed there.

Some unbounded-arity $\CNF_\Gamma$ can be recast as 
finite-arity $\exCNF_{\Gamma'}$.
In particular, Horn CNF is equivalent in expressive power to
$\exCNF_{\Gamma_{\mathrm{Horn}}}$ for the finite template
$\Gamma_{\mathrm{Horn}}=\{x\wedge y\to z,\ x\to y,\ x,\ \neg x\}$. To see this, note that for $k\ge2$ the Horn clause
$x_1\wedge\cdots\wedge x_k\to y$ is equivalent to
\[
  \exists z_1\cdots\exists z_{k-1}\ (x_1\wedge x_2\to z_1)
  \wedge\bigwedge_{i=2}^{k-1}(z_{i-1}\wedge x_{i+1}\to z_i)
  \wedge(z_{k-1}\to y),
\]
the constraints on the $z_i$ being themselves Horn, so that the least choice of
$z_1,\ldots,z_{k-1}$ makes $z_{k-1}$ equal to $x_1\wedge\cdots\wedge x_k$ and the
last conjunct is then exactly the clause.  Clauses with $k\le1$ are already
atoms, and a clause with no positive literal is treated the same way, with
$\neg z_{k-1}$ in place of $z_{k-1}\to y$.
Since $\Gamma_{\mathrm{Horn}}$ contains the clause $\bar x\vee\bar y\vee z$, it
falls under case~(b)~\cite[Theorem~11]{DalmauJeavons}: $\exCNF_{\Gamma_{\mathrm{Horn}}}$
is not polynomially PAC predictable with membership queries under the
cryptographic assumption above.  Horn CNF itself, by contrast, is polynomially properly PAC learnable
with membership queries~\cite{AFP}, though from random examples alone it is
not properly PAC learnable unless $\NP=\RP$~\cite{KLPV,ABFKP}, as discussed
above, and whether it is polynomially PAC predictable from random examples alone
is open, a PAC predictor settling the DNF problem by the same reduction.
There is no conflict between the two statements: a Horn CNF of $m$ clauses
translates into an $\exCNF_{\Gamma_{\mathrm{Horn}}}$-formula of $O(mn)$ atoms,
but the converse translation can be exponential --- existential quantification
acts on Horn clauses as fan-out does on circuits --- and hardness of the more
succinct class says nothing about the less succinct one.  A similar situation
holds for monotone CNF.

\section*{Open questions}

Several questions are left open above.  For the affine interval
$\{\oplus^3\}\preceq O\preceq\{\oplus,\top,\bot\}$, Theorem~\ref{thm:occam} gives
an Occam algorithm whose bound depends on $n$, and we do not know whether the
support of the fitted parity can be controlled in terms of the size of the
smallest fitting formula, that is, whether the algorithm can be made
attribute-efficient; nor, for the same interval, whether $\PL_O$ is PAC learnable
under random classification noise in polynomial time, which is the learning
parity with noise problem (Section~\ref{sec:noise}).  And for fragments given by a finite constraint language we do not
know where the line runs for empirical risk minimization, either for exact
solvability or for approximability.

\bibliographystyle{plain}
\bibliography{boolean_clone_dichotomies}

\end{document}